\documentclass[12pt]{iopart}

\usepackage{iopams}
\usepackage{graphicx}% Include figure files
\usepackage{dcolumn}% Align table columns on decimal point
\usepackage{bm}% bold math
\usepackage[utf8]{inputenc}

\expandafter\let\csname equation*\endcsname\relax
\expandafter\let\csname endequation*\endcsname\relax
\usepackage{amsmath}
\usepackage{amsthm}
\usepackage{appendix}
\usepackage{subfigure}
\usepackage{cite}

\newtheorem{theorem}{Theorem}
\newtheorem{definition}{Definition}

\newtheorem{proposition}{Proposition}

\newtheorem{corollary}{Corollary}

\newcommand{\Pauli}{\mathcal{P}}
\newcommand{\U}{\mathcal{U}}
\newcommand{\Q}{\mathcal{Q}}
\newcommand{\Pcal}{\mathcal{P}}

\newcommand{\supp}{{\rm supp}}
\newcommand{\trr}{\triangleright}
\newcommand{\K}{\mathcal{K}}
\newcommand{\E}{\mathbb{E}}
\newcommand{\Var}{{\rm Var}}
\newcommand{\ii}{{\rm i}}
\newcommand{\Cl}{{\rm Cl}}
\newcommand{\C}{\mathbb{C}}
\newcommand{\M}{\mathcal{M}}
\newcommand{\Stab}{{\rm Stab}}
\newcommand{\CNOT}{{\rm CNOT}}
\newcommand{\vabs}[1]{\left\| #1 \right\|}

\newcommand{\inv}{{\rm inv}}
\newcommand{\symp}{{\rm symp}}
\usepackage{hyperref}
\usepackage{float}
\hypersetup{colorlinks=true, linkcolor=blue, citecolor=blue, urlcolor=black}

\newcommand{\pbra}[1]{\left( #1 \right)}
\newcommand{\sbra}[1]{\left[ #1 \right]}
\newcommand{\cbra}[1]{\left\{ #1 \right\}}

 \newcommand{\ket}[1]{| #1 \rangle}
 \newcommand{\bra}[1]{\langle #1 |}
 \newcommand{\trace}{{\rm tr}}
 \newcommand{\abs}[1]{|#1|}

\newcommand{\rev}[1]{{\color{black} #1}}

\begin{document}

\title{Efficient measurement schemes for bosonic systems}

\author{Tianren Gu$^{1,2}$, Xiao Yuan$^{1,3}$ and Bujiao Wu$^{1,3}$
}
\address{$^1$Center on Frontiers of Computing Studies, Peking University, Beijing 100871, China}
\address{$^2$D-ITET, ETH Zürich, 8092 Zürich, Switzerland}
\address{$^3$School of Computer Science, Peking University, Beijing 100871, China}
\ead{bujiaowu@gmail.com}

\date{\today}

\begin{abstract}

Boson is one of the most basic types of particles  and preserves the commutation relation. 
An efficient way to measure a bosonic system is important not only for  simulating complex physics phenomena of bosons (such as nuclei) on a qubit based quantum computer, but for extracting classical information from a quantum simulator/computer that itself is built with bosons (such as a continuous variable quantum computer).  
Extending the recently proposed measurement schemes for qubits, such as shadow tomography and other local measurement schemes, here we study efficient measurement approaches for bosonic systems.
We consider truncated qudit and continuous variable systems, corresponding to simulated bosons on a discrete quantum computer and an inherent boson system, respectively, and propose different measurement schemes with theoretical analyses of the variances for these two cases. 
We numerically test the schemes for measuring nuclei vibrations simulated using a discrete quantum computer and a continuous variable Gaussian state, and the simulation results show great improvement of the performance of the proposed method compared to conventional ones. 
    
\end{abstract}

\noindent{\it Keywords\/}: bosonic system, measurement, qudit, continuous variables

\maketitle

\section{Introduction}

As one of the most basic types of indistinguishable particles in nature, boson preserves the commutation relation and describes a large set of basic particles, including photons, the W and Z bosons, gluons, Higgs bosons, etc, and composite particles, such as mesons and stable nuclei~\cite{braunstein2005quantum,weedbrook2012gaussian,pfister2019continuous}. 
An efficient way to simulate bosonic systems could thus enable us to study different quantum physics properties involving different types of forces. 
However, in contrast to fermions, bosons could stay in the same mode and thus the dimension could in principle go to infinity. 
When the excitation is low, such as the nuclei, we can truncate the dimension and effectively simulate bosons using qudits, which is implementable with a qudit-based quantum computer~\cite{blok2021quantum,chi2022programmable}. 
On the other hand, the excitation could also be high, such as continuous variable optics, we need to change to the continuous variable language~\cite{fukui2022building}. 
We can simulate it using a continuous variable device, which notably could also realize universal quantum computing~\cite{braunstein2005quantum,weedbrook2012gaussian}. 

A basic task in quantum information and computation is measuring the quantum state to extract classical outcomes.  
For conventional multi-qubit states, a possible measurement approach is to utilize quantum state tomography (QST)~\cite{leonhardt1997measuring,cramer2010efficient,gross2010quantum,lvovsky2009continuous,thew2002qudit} to reconstruct the quantum state, which however requires exponential resources. Interestingly, focusing on linear properties of an unknown quantum state, Aaronson et al.~\cite{aaronson2019shadow} proposed the ``shadow tomography'' algorithm, which only requires a polynomial number of copies of the quantum states. Nevertheless, an exponential size of quantum circuits and exponential classical post-processing cost are still needed.
Furthermore, Huang et al.~\cite{huang2020predicting} proposed the classical shadow (CS) tomography method, aiming to more efficiently approximate linear observables of the quantum state.  
The CS algorithm performs a quantum circuit $U$ uniformly randomly picked from the (local) Clifford group on the prepared state and then measures the ultimate state on the computational basis with measured basis $\{\ket{b}\}$. 
Then, we can calculate the quantum channel describing the overall process with the design properties of the Clifford group. 
The CS algorithm is efficient since the variance can be bounded to be logarithmic to the number of observables $M$ and square to the ``shadow norm'' of the observables by the 3-design property of the Clifford group. 
However, these methods do not exploit any prior information of the observable, which could further be leveraged to improve the measurement efficiency.
For example, consider an observable  that has a decomposition $H = \sum_{j}\alpha_j Q^{(j)}$ with real numbers $\alpha_j$ and Pauli strings $Q_i$, 
we can further consider the commutation relation of $\{Q_i\}$ as well as the weights $\abs{\alpha_j}$ associated with the Pauli strings, corresponding to variants of the CS method~\cite{wu2002qubits,hadfield2021adaptive,hadfield2022measurements,huang2021efficient,miller2022hardware,ippoliti2023classical,Zhao21Fermionic,low2022classical,wan2022matchgate,bertoni2022shallow} with even better performances.

Now considering a bosonic system that is either described by a truncated qudit or a continuous variable system, we study efficient measurement schemes for bosonic systems. 
%\wbj{Application to the calculation of ground state in bosonic system.}
Considering qudit systems, since the Clifford group only forms a 2-design~\cite{webb2015clifford} instead of a 3-design for $d>2$, the CS method has no theoretical guarantee in the qudit system. It can be proved that the estimation is still unbiased, while the variance bound cannot be straightforwardly generalized from qubit to qudit systems since the proof requires the 3-design property.
Numerical results show that the estimation error with the CS method on qudit systems is relatively large.
Here we propose a local measurement method for qudit systems by measuring the General Gell-Mann Basis (GGB) with an optimized probability. We prove that its variance can be bounded to $d^{3k}\vabs{O}_{\infty}^2$ in qudit systems, where $k$ is the locality of the observable $O$. {Our numerical results also show great improvement compared to a direct generalization of the CS method.}

We also propose a local measurement method for continuous variable descriptions of bosonic systems. We prove that the variance of the estimation of $\trace\pbra{\rho O}$ for any $k$-local boson operator $O=\sum_j \alpha_j Q^{(j)}$ can be bounded by $3^k \sum_p \abs{\alpha_p}^2B^{2{\rm deg}(Q^{(j)})}$, where $\deg(Q^{(j)})$ is the degree of the decomposed monomial $Q^{(j)}$ using the continuous variables and $B$ is an upper bound of measurement result of $x$ and $p$ under properly chosen coordinate, $Q^{(j)}$ has tensor product form $ x_1^{l_1}p_1^{m_1} \otimes x_2^{l_2}p_2^{m_2} \otimes\cdots\otimes x_n^{l_n}p_n^{m_n}$, and $x,p$ denote position and momentum quadratures respectively.
We benchmark this method for continuous variable optics.

In the following, we  first review bosonic systems in continuous and discrete representations, the classical shadow method, and its variants in Sec.~\ref{sec:pre}. We then propose a local measurement method and analyze the CS and the local measurement method on qudit systems in Sec.~\ref{sec:BosonDiscrete}.
In Sec.~\ref{sec:BosonContinuous}, we propose a local measurement method for continuous variable systems. We then give numerical experiments in Sec.~\ref{sec:numerical}  and a discussion in Sec.~\ref{sec:discussion}.

\section{Preliminary}
\label{sec:pre}

In this section, we provide the basic concepts and notations about bosonic systems with continuous and discrete representations. We also review the CS algorithm~\cite{huang2020predicting}, and the unified representation of local measurements schemes~\cite{wu2021overlapped} for qubit systems.

\subsection{Bosonic systems and their representations}
For a system of $n$ bosonic modes, we usually describe it with corresponding creation and annihilation operators $b_1^\dagger,b_2^\dagger,\cdots,b_n^\dagger$ and $b_1,b_2,\cdots,b_n$. Since bosons could stay in the same mode, the dimension of a bosonic state could be infinity, making its representation hard~\cite{friedberg1989gap,weedbrook2012gaussian}. 
There are two approaches to overcome this issue by using (1) dimension truncation and (2) continuous variables. 

Dimension truncation is a simple yet useful method to represent bosonic systems~\cite{thew2002qudit,chi2022programmable,blok2021quantum}. 
In many practical models, we only have low excitations and hence a smaller number of occupations in the same mode.
For such cases, we can neglect the highly excited states with more than a fixed number of $d$ bosons. Therefore the $n$ bosonic mode system becomes a corresponding $n$ qudit system. Similar to qubit systems, we can define Pauli operators and the Pauli group for qudit systems, as shown in Definition~\ref{def:qudit}.

\begin{definition}
\label{def:qudit}
For the orthonormal computational basis $\{\ket{s}|s \in\cbra{ 0,1,\cdots,d-1}\}$ on $d$-dimensional Hilbert space, the Pauli operators $X_d$ and $Z_d$ are defined as
\begin{equation}
    X_d\ket{s} = \ket{s\oplus 1}, \quad Z_d\ket{s} = w^s \ket{s},
\end{equation}
where $\omega= e^{2\pi\ii s /d}$. The corresponding single-qudit Pauli group $\Pauli_d$ is defined as $\Pauli_d = \langle \omega I_d, X_d, Z_d\rangle$, and the $n$-qudit Pauli group is defined as the tensor products of single-qudit $d$-dimensional Pauli group
\begin{equation}
    \Pauli_d^n = \Pauli_d^{\otimes n}.
\end{equation}
\end{definition}

These qudit Pauli operators are generally not Hermitian. Hence they can not be directly utilized as local observables as in the qubit scenario. Yet, the particle number operator $N = b^\dagger b$ is compatible with the operator $Z_d$, enabling us to estimate $Z_d$ via the measurement of $N$. The measurement of other qudit Pauli operators can be converted to the measurement of $Z_d$ through unitary transformations.

Another way is to use continuous variables to represent bosons~\cite{braunstein2005quantum,pfister2019continuous}.
Usually, a bosonic mode can be modeled as an oscillator with the position $x$ and momentum $p$ quadratures
\begin{equation}
    x = b + b^\dagger,\ p = -i(b-b^\dagger),
    \label{eq:cv_transformation}
\end{equation}
Note that in contrast to the Majorana operators in fermionic systems~\cite{zhao2021fermionic},  they neither commute nor anti-commute. The $x$ and $p$ operators describe a virtual oscillator with the Hamiltonian
\begin{equation}
    H = \frac{1}{2}p^2 + \frac{1}{2}x^2.
    \label{eq:boson_single_hamiltonian}
\end{equation}
The quantum state with $n$ excitations (bosons) of the mode corresponds to the $n$th excited eigenstate of the Hamiltonian $H$.
A general bosonic system could involve multiple modes.
In practice, we are interested to measure observables that are written in the form of creation and annihilation operators, which can always be compiled to the $x$ and $p$ operators. 

\subsection{Classical shadows}
Next, we briefly review the classical shadow algorithm introduced by Huang et al.~\cite{huang2020predicting}. Classical shadows are classical estimators of a general quantum state $\rho$ which can simultaneously predict $M$ linear functions $\trace(O_1\rho),\cdots,\trace(O_M\rho)$ with few measurements of the quantum state\rev{, where observables $O_1, \cdots, O_M$ are Hermitian}. Obtaining a classical shadow requires a copy of the quantum state $\rho$ and a `twirling' operation $U$ that maps
\begin{equation}
    \rho \mapsto U\rho U^\dagger,
\end{equation}
where $U$ is randomly chosen from some unitary ensemble $\U$. 
The algorithm measures the rotated state $U\rho U^\dagger$ under a set of computation basis $\{\ket{z},z\in\{0,1\}^n\}$. After obtaining the measurement result $b$, 
an estimation of $\rho$ can then be represented as
\begin{equation}
    \hat\rho = \M^{-1}(U^\dagger\ket{b}\bra{b}U),
\end{equation}
which is called the classical shadow of the state $\rho$, where the shadow channel $\M$ is defined as
\begin{equation}
    \M(\rho) = \mathop{\E}_{U\in \U}\sum_{b\in\{0,1\}^n} \bra{b} U\rho U^\dagger \ket{b} U^\dagger \ket{b}\bra{b} U = \mathop{\E}_{U\in\U,b\in\{0,1\}^n} \sbra{U^\dagger \ket{b}\bra{b} U}.
    \label{definition_shadow_channel}
\end{equation}
It is proved that $\hat{\rho}$ is an unbiased estimator for $\rho$ and the error of the estimation to linear function $\trace(O_i\rho)$ can be bounded using the shadow norm
\begin{equation}
    ||O||_{\rm shadow}^2 := \max_{\sigma:state} \left(
    \E_{U\in \U,b\in\{0,1\}^n} \bra{b}U \M^{-1}(O) U^\dagger\ket{b}^2
    \right).
\end{equation}
In particular, the number of samples needed to estimate all $M$ linear functions up to an additive error $\epsilon$ is given by
\begin{equation}
     N_{\rm sample} = \mathcal{O}\left(\frac{\log M}{\epsilon^2}||O||_{\rm shadow}^2\right).
\end{equation}
For the qubit Clifford ensemble, the shadow norm for observable $O$ is bounded by $3\trace(O^2)$ while for the qubit Pauli ensemble, the shadow norm is 
\rev{$4^k||O||_\infty^2$, where $k$ denotes the locality of $O$, i.e.~the maximum number of non-identity Pauli operators in the Pauli string decomposition of $O$}.

\subsection{Local measurements in qubit system}
\label{section_local_measurements}
While the classical shadow method works for general observables, it also neglects the information of the observables. Here, we show an alternative more general framework proposed by Wu et al.~\cite{wu2021overlapped} that unites classical shadow and other measurement schemes, including importance sampling~\cite{wecker2015progress,mcclean2016theory}, grouping method~\cite{kandala2017hardware,verteletskyi2020measurement,vallury2020quantum,izmaylov2019unitary,zhao2020measurement,hempel2018quantum,o2016scalable,crawford2021efficient}, (local-biased) classical shadow (CS)~\cite{hadfield2021adaptive,huang2020predicting}, and overlapped grouping measurement (OGM) algorithm~\cite{wu2021overlapped}.
We only consider the Pauli measurements $P\in \Q := \cbra{I, X, Y, Z}^{\otimes n}$ in qubit systems.  
For a given observable $O$ with a decomposition $O = \sum_j \alpha_j Q^{(j)}$, where $Q^{(j)}\in \Q$ and a given quantum state $\rho$, we can estimate $\trace\pbra{\rho Q^{(j)}}$ by measuring with Pauli-strings in $\Q$.
Any local measurement $P\in\Q$ can be written as a tensor product $P = P_1\otimes P_2\otimes\cdots\otimes P_n$. Measuring the quantum state with $P$ (eigenstate of $P$) gives us an $n$-bit string $\mu(P)=\mu(P_1) \ldots\mu(P_n)\in \cbra{\pm 1}^n$.

It can be easily checked that we can estimate $\tr\pbra{\rho Q^{(j)}}$ with the measurement $P\in\Q$ if and only if $P$ covers all components of $Q^{(j)}$, i.e.~$P_i = Q^{(j)}_i$ for all $i\in\supp\pbra{Q^{(j)}}$, where $\supp(Q^{(j)})=\{i|Q^{(j)}_i \ne I\}$ denotes the index set of the non-identity components in $Q^{(j)}$, which we denote as $Q^{(j)}\trr P$.
With these notations, we define the estimation of $\tr\pbra{\rho Q^{(j)}}$ as
\begin{equation}
    \mu(P,Q^{(j)}) = \begin{cases}
    \prod_{i\in \supp(Q^{(j)})}\mu(P_i), &\text{ if } Q^{(j)}\triangleright P\\
    0, &\text{ otherwise.}
    \end{cases}
    \label{estimator_Qj}
\end{equation}
 An estimator of the expectation of Hamiltonian $O$ is thus given by
\begin{equation}
    \hat o = \sum_j \alpha_j \mu(P,Q^{(j)}).
    \label{estimator_O_single}
\end{equation}

For multiple measurements, we further consider the distribution of measurements $\K(P)$ to optimize the performance under a fixed number of measurements and rewrite estimator (\ref{estimator_O_single}) as

\begin{equation}
    \hat o = \sum_j \alpha_j f(P,Q^{(j)},\K)\mu(P,Q^{(j)}),
    \label{estimator_O}
\end{equation}
\rev{
where $\K$ is the distribution of measurements $P$, $f(P,Q^{(j)},\K)$ is a function satisfying $\E_{P\sim \K,Q^{(j)}\trr P}[f(P,Q^{(j)},\K)] = 1$ for every $Q^{(j)}$. Since $\mu(P,Q^{(j)})$ is an unbiased estimator of $\trace(\rho Q^{(j)})$, the estimator Eq.~\eqref{estimator_O} satisfies
\begin{equation}
\begin{split}
    \E[v] = \E_{P\sim\K}\E_{\mu(P)}[v] &= \sum_j \alpha_j \E_{P\sim\K}[f(P,Q^{(j)},\K)\E_{\mu(P|Q^{(j)})}\mu(P,Q^{(j)})]\\ 
    &=\sum_j \alpha_j \E_{P\sim\K, Q\trr P}[f(P,Q^{(j)},\K)]\trace(\rho Q^{(j)}) \\
    &= \sum_j \alpha_j \trace(\rho Q^{(j)})\\
    &= \trace(\rho O).
\end{split}
\end{equation}
The design of the function $f(P,Q^{(j)},\K)$ influences the variance of the estimator, and hence determines the error of the estimators, as shown in Ref. \cite{wu2021overlapped}.
The variance of the estimator defined in Eq.~\eqref{estimator_O} is given by
}
\begin{equation}
    \Var[\hat o] = \sum_{j,k}\left[\alpha_j\alpha_k g(Q^{(j)},Q^{(k)})\trace(\rho Q^{(j)}Q^{(k)})\right] - [\trace(O\rho)]^2,
    \label{variance_O}
\end{equation}
\rev{
where $g(Q^{(j)},Q^{(k)}) = \E_{P\sim\K,Q^{(j)}\trr P,Q^{(k)}\trr P} [f(P,Q^{(j)},\K)f(P,Q^{(k)},\K)]$ uniquely depends on the expected value over the distribution $\K$ and the factor function $f$.  Once the specific form of $f(P, Q^{(j)}, \K)$ is established, our goal is to optimize the distribution $\K(P)$ since the variance described in Eq. \eqref{variance_O} is merely a linear combination of $g(Q^{(j)}, Q^{(k)})$. By following this approach, it is possible to replicate most existing local measurement optimization algorithms. For instance, in qubit systems, when $\K(P) = 1/3^n$ for all possible measurements $P$ and $f_{\rm CS}(P,Q^{(j)},\K) = \prod_i \left(\delta_{Q_i^{(j)},I} + 3\delta_{Q_i^{(j)},P_i}\right)$, the algorithm corresponds to CS. Similarly, when $\K(P^{(j)}) = |\alpha_j|/||\bm{\alpha}||_1$ for $P^{(j)} = Q^{(j)}$ and $f_{l_1}(P,Q^{(j)},\K) = \K(P)^{-1}\delta_{P,Q^{(j)}}$, the algorithm becomes importance sampling.
}

\section{Measuring bosonic systems (qudits)}
\label{sec:BosonDiscrete}
In this section we consider the case that the bosonic system is described by qudits. 
Here we give a local measurement method to qudit systems with measurements chosen from General Gell-Mann Basis in subsection \ref{subsec:localmeas_qudit}. We additionally prove that the variance of the estimation of  $\tr\pbra{\rho O}$ on truncated qudit systems can be bounded.  As a comparison, we also give the estimator with the Clifford group in qudit systems in subsection \ref{subsec:cs_qudit}.

\subsection{Qudit local measurements}
\label{subsec:localmeas_qudit}
To introduce the local measurement scheme to the qudits, we first introduce the General Gell-Mann Basis (GGB)~\cite{griffiths2020introduction,thew2002qudit,bertlmann2008bloch} as follows
\begin{equation}
\begin{aligned}
&\Lambda_s^{jk} = \ket{j}\bra{k} + \ket{k}\bra{j}\\
&\Lambda_a^{jk} = -i\ket{j}\bra{k} + i\ket{k}\bra{j}\\
&\Lambda^l = \sqrt{\frac{2}{l(l+1)}}\left(\sum_{j=0}^l \ket{j}\bra{j} -l\ket{l+1}\bra{l+1}\right).
\end{aligned}
\label{gellmann}
\end{equation}
\rev{A $d$-dimensional GGB basis consists of $d(d-1)/2$ symmetric and asymmetric matrices indexed by $j,k\in \cbra{0,1,\cdots,d-1}$, as well as $d$ diagonal matrices indexed by $l\in \cbra{0,1,\cdots,d-1}$. For simplicity, we reindex these matrices as $\{\Lambda_j\}$ with subscript $j \in\cbra{ 0,1,\cdots, d^2-1}$.}
These matrices are orthonormal since $\trace\pbra{\Lambda_j^\dagger\Lambda_k}=2\delta_{jk}$ and tomographically complete, enabling us to use them as the measurement basis. We then define the GGB string on multiple qudits as 
\begin{equation}
Q^{(j)} = \Lambda_{j_1}\otimes\Lambda_{j_2}\otimes\cdots\otimes\Lambda_{j_n},
\end{equation}
with which we can apply the local measurement scheme.
\rev{Similar to the Pauli strings, each $n$-mode Hermitian observable $O$ can be decomposed as a linear combination of $n$-term GGB strings, which is shown as follows,
\begin{equation}
    O = \sum_j \alpha_j Q^{(j)}, \quad j\in \{0,1,\cdots,d^2-1\}^{\otimes n},
    \label{eq:GGB_decompostion}
\end{equation}
for some real coefficients $\alpha_j$.
}

We can apply the local measurement methods as in qubit systems, as shown in Sec.~\ref{section_local_measurements} to the GGB of qudit systems. \rev{An observable $O$ is called $k$-local if and only if each GGB string in decomposition \eqref{eq:GGB_decompostion} has maximally $k$ non-identity terms. $||\cdot||_\infty$ denotes the spectral norm, we also refer to $k$ as the locality of $O$.}
With the properties of GGB, we can bound the variance of the estimations with the local measurement methods as follows. 

\begin{proposition}
Let $\hat o$ defined in Eq. \eqref{estimator_O} be the estimation of $\tr\pbra{\rho O}$ under the classical shadow scheme (uniform measurement probability) of GGB, then the variance of the estimator $\hat o$ can be bounded to
\rev{

\begin{equation}
    {\rm Var}(\hat o) \le d^{2k}(d-1)^k||O||_\infty^2,
\end{equation}
\label{pro:qudit_local}
where $k$ denotes the locality of the observable $O$.

}
\end{proposition}

To prove this bound, we need to calculate the function $g(Q_p, Q_q)$  given  in Eq.~\eqref{variance_O}. For the uniform measurement probability, the task can be simplified to count the number of measurements $P$ such that $Q_p \triangleright P$ or $Q_q\triangleright P$ or both. Setting $|Q|$ as the non-identity terms in GGB string $Q$, we can thus figure out $g(Q_p,Q_q) = (d^2-1)^s$ for $s = |Q_p\bigcap Q_q|=|Q_p|+|Q_q|-|Q_p\bigcup Q_q|$. Using this result together with the Cauchy-Schwarz inequality, we can prove Proposition~\ref{pro:qudit_local}. 
See \ref{appendix_variance} for the proof details, which is a generalized version of the variance bound of local measurement schemes under uniform measurement probability.

\subsection{Classical shadow for qudits}
\label{subsec:cs_qudit}
Here we also introduce the generalization of the classical shadow method to qudit systems. We show that the estimation from a uniformly sampled global Clifford group in qudit systems is still an unbiased estimation of $\tr\pbra{\rho O}$, while it is hard to bound the variance of the estimation.

To apply the classical shadow scheme to the qudit space, we first introduce the following qudit Clifford group.
\begin{definition}
The $d$-dimensional $n$-mode Clifford group is defined as the normalizer of the Pauli group $\Pcal_d^n$ on the qudit system:
\begin{equation}
    \Cl(d^n) = \{c\in U(d^n)| c\Pauli_d^n c^\dagger = \Pauli_d^n\},
\end{equation}
\end{definition}
When $d$ is a prime, the $d$-dimensional Clifford circuit can be efficiently simulated on a classical computer, which is an extension of the Gottesman-Knill Theorem (For details see \ref{dCliffordsimulation}). Moreover, it can be proved that the $d$-dimensional Clifford group forms a unitary 2-design~\cite{webb2015clifford}. Then we can apply the integration in \cite{gross2015partial}
\begin{equation}
 \E_{\Cl(d^n)} \bra{x}UAU^\dagger\ket{x} U^\dagger\ket{x}\bra{x}U = \frac{A + \trace(A)I}{d^n(d^n+1)},\quad \forall A \in \mathbb{H}(\C^{d^n}).
\end{equation}
Therefore, the shadow channel with the qudit Clifford group is
\begin{equation}
    \M_{\Cl(d^n)}(\rho) = \frac{\rho + I}{d^n+1},
\end{equation}
and the corresponding classical shadow is
\begin{equation}
    \hat{\rho} = (d^n+1)U^\dagger\ket{b}\bra{b}U - I_d,
\end{equation}
where $U\in \Cl(d^n)$ and $b\in \{0,1,\cdots,d^n-1\}$.\\

However, as the qudit Clifford group generally does not yield a 3-design like a qubit Clifford group, the variance of the prediction cannot  be directly estimated.

\section{Measuring bosonic systems (continuous variables)}
\label{sec:BosonContinuous}

Here we consider the continuous variable description of bosons. In particular, we assume that the target observable is a linear combination of products of the $x$ and $p$ operators. We show how to measure these general types of observables efficiently. 

\subsection{p-x strings}
We consider $\textbf{p}$ and $\textbf{x}$ as the local measurements. For a general one-mode boson observable $O = O(p,x)$, a natural basis is $x^lp^m$ with $l,m \ge 0$. Thus for $n$-mode bosonic systems, we can define the so-called p-x strings similar to Pauli strings. A p-x string can be generally written as the tensor product form $Q = x_1^{l_1}p_1^{m_1} \otimes x_2^{l_2}p_2^{m_2} \otimes\cdots\otimes x_n^{l_n}p_n^{m_n}$. We define the degree ${\rm deg}(Q)$ of a p-x string as the summation of its exponents such that ${\rm deg}(Q) = \sum_{j=1}^n (l_j+m_j)$. \rev{We denote by $\mathcal{Q}_{px}$ the set of all possible $n$-mode p-x strings, i.e.~$\mathcal{Q}_{px} = \{x_1^{l_1}p_1^{m_1} \otimes x_2^{l_2}p_2^{m_2} \otimes\cdots\otimes x_n^{l_n}p_n^{m_n}| l_i,m_i\in \mathbb{N}_+,i=1,\cdots,n\}$}. It can be proved that $\Q_{px}$ form a basis for any $n$-mode bosonic operators
%\revd{All possible p-x strings form a basis $\Q_{px}$ for any  $n$-mode boson operators}
\begin{equation}
O = O(\textbf{p},\textbf{x}) = \rev{\sum_{j=1}^M} \alpha_j Q^{(j)},\quad Q^{(j)} \in \Q_{px}
\label{O_px_decomposition}
\end{equation}
We note that operators in this basis are not orthogonal to each other. Commonly we may use a Gram-Schmidt process to obtain a series of orthonormal polynomials and then decompose the observables under this basis, yet in the actual case the observables we care about in many physical models are only functions of the position and the momentum, so a simple expansion would give us the form \eqref{O_px_decomposition}.

Generally, the possible measurement basis of the p-x strings in the expansion form of observable $O$ may still be infinite as $x$ and $p$ do not commute to each other. For such observables, we can apply an importance sampling and truncate the negligible terms. Here, among all p-x strings, we consider a special case where each term in the p-x string should be $x^l$ or $p^m$, i.e. $l_j\cdot m_j = 0$ for all $j\in\{1,\cdots,n\}$ so cross-terms such as $x_1p_1$ do not exist, called \emph{pure p-x strings}. Pure p-x strings correspond to physical systems whose Hamiltonian or other quantities are separable in each mode. In this case, each term of the p-x string either commutes with $x$ or commutes with $p$. We can therefore measure it with a simple p-x string with each term being $x$ or $p$, which gives us a limited number of possible measurements.

\subsection{Continuous variable local measurements}
\label{subsec_outmethod_main}
After constructing the p-x strings and decomposing the observables with the p-x string basis, we can use the measurement scheme in Section \ref{section_local_measurements}. The estimator (\ref{estimator_O}) and the general variance (\ref{variance_O}) remain unchanged when we use $\Q_{px}$ as the measurement basis.

Generally, we can hardly estimate the efficiency of p-x string measurements as the cardinality of $\Q_{px}$ is infinite. However, in the special case of pure p-x strings, we can bound the variance \eqref{variance_O} since the number of possible measurements is finite. 
\rev{We define the \emph{spectral projection} of the Hermitian operator $\hat A$ associated with the interval $[-B,B]$ as
\begin{equation}
1_B(\hat A) = \sum_{-B\le a\le B}\ket{a}\bra{a},
\end{equation}
where $\hat A$ has spectral decomposition $\hat A=\sum_a  a\ket{a}\bra{a}$ and $B>0$. Here the spectral projection replaces continuous spectral with the summation with integration.
To give an accurate description for the overall error, we introduce a set of positive constants $\cbra{B_i^{(j)}}_{i\in [n]}^{j\in [M]}$, which satisfy
\begin{equation}
    |\Tr\left[Q^{(j)}(\rho - P^{(j)}\rho P^{(j)})\right]| < \varepsilon_B,
    \label{eq:Bij_bound}
\end{equation}
for $j\in [M]$, where
\begin{equation}
P^{(j)} = \prod_{Q_i^{(j)}\ne I}1_{B_i^{(j)}}(Q_i^{(j)}),
\label{eq:Bij_projector}
\end{equation}
being projectors mapping the state to a finite range in position space or momentum space depending on the corresponding term in $Q^{(j)}$.
\begin{theorem}
Let $O$ be an $n$-mode Hermitian observable whose p-x string decomposition \eqref{O_px_decomposition} only contains pure p-x strings, and $\rho$ be an $n$-mode bosonic state satisfying $\abs{\tr(\rho Q^{(j)})} < \infty$ for $1\le j\le M$. The prediction error for the estimator defined in Eq.~\eqref{estimator_O} with {uniformly distributed}  measurement $\cbra{P^{(j)}}_{j=1}^M$ defined in Eq.~\eqref{eq:Bij_projector} can be bounded as $\varepsilon_o + \varepsilon_B \sum_{j=1}^M |\alpha_j|$ with probability at least $1-\delta$ with the number of measurements
\begin{align}
3^k\sum_j |\alpha_j|^2\prod_{i=1}^n(B_i^{(j)})^{2l_i^{(j)}}/(\varepsilon_o^2\delta).
\label{eq:thm_var_px}
\end{align}
\label{thm:var_px}
\end{theorem}
\begin{proof}
Summing over $j$ in Eq.~\eqref{eq:Bij_bound} we get the bias bound $\varepsilon_B\sum_{j=1}^M|\alpha_j|$. Noted that by Eq.~\eqref{eq:Bij_projector} the projected spectral norm $||P^{(j)}Q^{(j)}P^{(j)}||_\infty = \prod_{i=1}^n(B_i^{(j)})^{l_i^{(j)}}$, referring to \ref{appendix_variance} the variance of estimator \eqref{estimator_O} equals to $3^k\sum_j |\alpha_j|^2\prod_{i=1}^n(B_i^{(j)})^{2l_i^{(j)}}$. Eq.~\eqref{eq:thm_var_px} then follows. 
\end{proof}
}

\rev{
In Theorem \ref{thm:var_px} the condition $\abs{\tr(\rho Q^{(j)})}<\infty$ assures that the wave function decreases even faster than the p-x string with the largest exponent in the decomposition \eqref{O_px_decomposition}. Therefore, as $B_i^{(j)}\to\infty$ for all $i$ and $j$, we have $\varepsilon\to 0$, indicating that an infinite projection range gives an unbiased estimator. 

The constants $\cbra{B_i^{(j)}}_{i\in [n]}^{j\in [M]}$ can be seen as the detection range of the device, i.e., we suppose all the measurement outcomes would be bounded by $B_i^{(j)}$ and discard the outliers. This method inevitably causes bias but reduces variance, hence there should be a trade-off. However, in any real physical system, the 
wave function of the quantum state $\rho$ is in the Schwartz space $\mathcal{S}(\mathbb{R}^{2n})$, i.e., the space of rapidly decreasing functions. For example, most bosonic systems have a Gaussian-like wave function, which decreases at a rate of $e^{-x^2}$. In such cases, the value of $\varepsilon$ diminishes rapidly as we progressively increase $\cbra{B_i^{(j)}}_{i\in [n]}^{j\in [M]}$. Consequently, we can safely neglect the bias and approximately regard the estimator \eqref{estimator_O} as unbiased.
%\\ \newline

Theoretically we can calculate the infimum of each $B_i^{(j)}$, which can be proved as the tight bound. However, since such calculation is usually hard, 
it is customary to employ a shared bound denoted as $B$, rather than utilizing a set of $\cbra{B_i^{(j)}}_{i\in [n]}^{j\in [M]}$. As a consequence, Theorem \ref{thm:var_px} can be transformed into the subsequent corollary.

\begin{corollary}
    For any $k$-local operator $O$ that its p-x string decomposition in Eq. \eqref{O_px_decomposition} which only contains pure p-x strings, the variance of the estimator as defined in Eq. \eqref{estimator_O} under the classical shadow scheme (uniform measurement probability) is bounded by
\begin{equation}
{\rm Var}[\hat o] < 3^kB^{2K}\sum_p|\alpha_p|^2,
\label{Var_o_px}
\end{equation}
where $K$ is the maximum degree of the decomposed p-x strings and $B$ is a positive constant denoting the projection range of both $x$ and $p$. 
\label{p-x theorem}
\end{corollary}
Given bias parameter $\varepsilon$, the infimum of $B$ that Eq.~\eqref{Var_o_px} strictly holds for arbitrary observable $O$ and state $\rho$ is generally hard to compute. 
\rev{In practical scenarios, we carefully select an appropriate value for $B$ that serves as an upper bound for $\cbra{B_i^{(j)}}_{i\in [n]}^{j\in [M]}$}. Usually, we set $B$ as the range that contains a cutting probability of finding the particle.
To illustrate, when considering a Gaussian state characterized by a zero mean and standard deviation $\sigma$, it is advisable to consider setting the value of $B$ as $3\sigma$. Although this bound may not be optimal in all scenarios, it generally yields satisfactory performance.
}

To further improve the performance, we introduce the idea of overlapped grouping measurement~\cite{wu2021overlapped} to our continuous variable scheme, where we maximally group the $Q^{(j)}$s that are compatible (commute on each mode) and assign a corresponding measurement $P$ for each group. In this case the function $f_G(P,Q,\K) = \left(\sum_{P:Q\trr P}\K(P)\right)^{-1}\delta_{Q\trr P}$ for some given measurement distribution $\K(P)$. We can then optimize the variance by choosing the distribution $\K(P)$ to minimize the diagonal variance
\begin{equation}
    l(\K) = \sum_j \frac{\alpha_j^2}{\sum_{P:Q^{(j)}\triangleright P}\K(P)},
    \label{eq:cost_func}
\end{equation}
which is state invariant and can be evaluated much faster than the actual variance. 
We use Eq.~\eqref{eq:cost_func} as the main optimization method for the numerical experiments of continuous variable local measurements in the following sections. \\

In the special case that the observable $O$ can be separated to the simple summation of momentum functions $U(\bm{p}) = U(p_1,\cdots,p_n)$ and position functions $V(\bm{x}) = V(x_1,\cdots,x_n)$, the variance bound can be further improved to
\begin{equation}
    \Var(\hat o) \le (||U||_\infty + ||V||_\infty)^2.
\end{equation}
The proof is given in \ref{separableobservable}. Particularly, for bosons in a potential field, whose Hamiltonian can be generally written as

\begin{equation}
H = \frac{1}{2}\sum_{i=1}^{N} p_i^2 + V(\bm{x}),
\label{Ham_px_N}
\end{equation}
the variance can be bounded by
\begin{equation}
{\rm Var}(\hat o) \le (\frac{N}{2}B_p + ||V||_\infty)^2,
\label{bosoninpotential}
\end{equation}
where $B_p$ is the bound of momentum. 
\rev{
Assuming that $V(\bm{x})>0$, and the total energy of the bosonic system is finite, i.e. $\tr(\rho H) < E$ for some $E>0$, a reasonable suggestion is to choose $B_p = \sqrt{2E}$.
}

\subsection{Measuring with noise}

Because of the precision problem, obtaining the exact measurement results of continuous variables is impossible.
In the above discussion, we did not consider the measurement noise for continuous variables. Here we give the revised estimation and variance bound for continuous variables with measurement noise.

Suppose that the measurement results $\mu(P)$ for measurement $P$ are shifted by some random error $e(P)$, as in the following equation
\begin{equation}
\mu(P)\to \mu'(P) = \mu(P) + e(P).
\end{equation}
As we only consider the precision of the measurement device, we assume these errors are small, independent of the quantum state, and independent of each other. Their expectations should be 0, i.e. $\mathbb{E}e(P_i) = 0$ for each possible $P_i$. Then the estimator (\ref{estimator_O}) becomes 
\begin{equation}
\hat o = \sum_j \alpha_j f_G(P_G,Q^{(j)},\mathcal{K})\mu' (P,{\rm supp}(Q^{(j)})).
\label{v_G_noise}
\end{equation}
With this revised definition of the estimation, the associated variance is given in Proposition \ref{pro:Var_noise}.
\begin{proposition}
\label{pro:Var_noise}
The estimator with noise defined in (\ref{v_G_noise}) is unbiased, and the variance is equal to
\begin{equation}
\begin{aligned}
{\rm Var}(\hat o) &= \sum_{j,k} \alpha_j\alpha_k g(Q^{(j)},Q^{(k)}) \left(\trace(\rho Q^{(j)} Q^{(k)})+\sum_{i\in {\rm supp}(Q^{(j)})\cap {\rm supp}(Q^{(k)})}\trace(\rho\frac{Q^{(j)}Q^{(k)}}{P_i^2}){\rm Var}(e(P_i))\right)\\ &= V_o + V_e,
\end{aligned}
\end{equation}
where $V_o$ is the original variance given in \eqref{variance_O} and bounded by \eqref{Var_o_px}, and $V_e$ is the extra variance caused by the noise, which is bounded by
\begin{equation}
    V_e \le nB^{2k-2}B_e^2\sum_p |\alpha_p|^2,
\end{equation}
where $B_e$ is the upper bound of the error $e(P)$.
\end{proposition}
\rev{Similar to the explanation of $B$ in Theorem \ref{thm:var_px} and Corollary \ref{p-x theorem}, the error bound $B_e$ can be viewed as the detection precision of the device. Usually we deal with a Gaussian error $e(P)$ with zero-mean, where $B_e$ can be chosen as the standard deviation.}

We prove Proposition \ref{pro:Var_noise} by expanding the noisy measurement outcome product $\mu'(P,\supp(Q)) = \prod_{i\in\supp(Q)} (\mu(P_i)+e(P_i))$ to the second term and analyse the corresponding expectation and variance. The detailed proof is given in \ref{measurewithnoise}.

\section{Numerical Experiments}
\label{sec:numerical}
In this section we numerically verify the theoretical results.
~\footnote{The source code can be found at \\\url{ https://github.com/Gutianren/Efficient-measurement-schemes-for-bosonic-systems}}.

\subsection{Discrete: Nuclei Vibrations}

Here we consider the computation of molecular vibronic spectra using a digital quantum computer. We compare the discrete methods shown 
in Sec.~\ref{sec:BosonDiscrete} and compare their efficiency. 
We consider the $\rm H_2O$ molecule with Hamiltonian representation as
\begin{equation}
    H_{\rm mol} = \sum_s \ket{s}\bra{s}_e \otimes H_s = \sum_s \ket{s}\bra{s}_e \otimes \left(-\sum_I \frac{\nabla^2_I}{2M_I} + V_s(\bm{R}_I)\right),
\end{equation}
under the Born-Oppenheimer Approximation~\cite{mcardle2019digital}\rev{, where $s$ indexes the electrons' eigenstates and $H_s$ is the corresponding nuclear Hamiltonian with $M_I$ being the reduced mass of the nuclei and $V_s(\bm{R}_I)$ being the potential field between the nuclei}. We can further simplify the Hamiltonian $H_s$ by working in mass-weighted normal coordinates, decoupling the vibrational modes, and simplifying the potential function $V_s$, which can then be written as
\begin{equation}
    H_s = \sum_i \frac{\bm{p}^2}{2} + V_s(\bm{q}) \approx \sum_i \frac{\bm{p}^2}{2} + \sum_i\frac{1}{2}\omega_iq_i^2 = \sum_i \omega_i a_i^\dagger a_i.
\end{equation}
Then the Hamiltonian could be regarded as an independent Harmonic oscillator. 
To improve the result, we can consider higher order expansions of $V_s(\bm{q})$ to involve interactions of the Harmonic oscillators  as
\begin{equation}
    V_s(\bm{q}) = \sum_{j=2}^\infty \sum_{i_1,\cdots,i_j=1}^M k_{i_1,\cdots,i_j}q_{i_1}\cdots q_{i_j}.
\end{equation}
To simulate the Hamiltonian, we truncate each mode to a $d$-level system, which provides approximate solutions to the low-lying eigen-energies. Specifically, we consider the lowest $d$-levels for each Harmonic oscillator and then map the Hamiltonian to this subspace.
Consider the ${\rm H_2O}$ molecule vibration model $H_\text{mol}$, 
we numerically implement the qudit local measurement with overlapped grouping optimization and CS for qudits algorithms to calculate the expectation of $H_\text{mol}$ under the oscillator's ground state, and $d$-level GHZ state respectively. 

\begin{figure}[ht]
    \centering
    \subfigure[]{          
        \includegraphics[width=7cm]{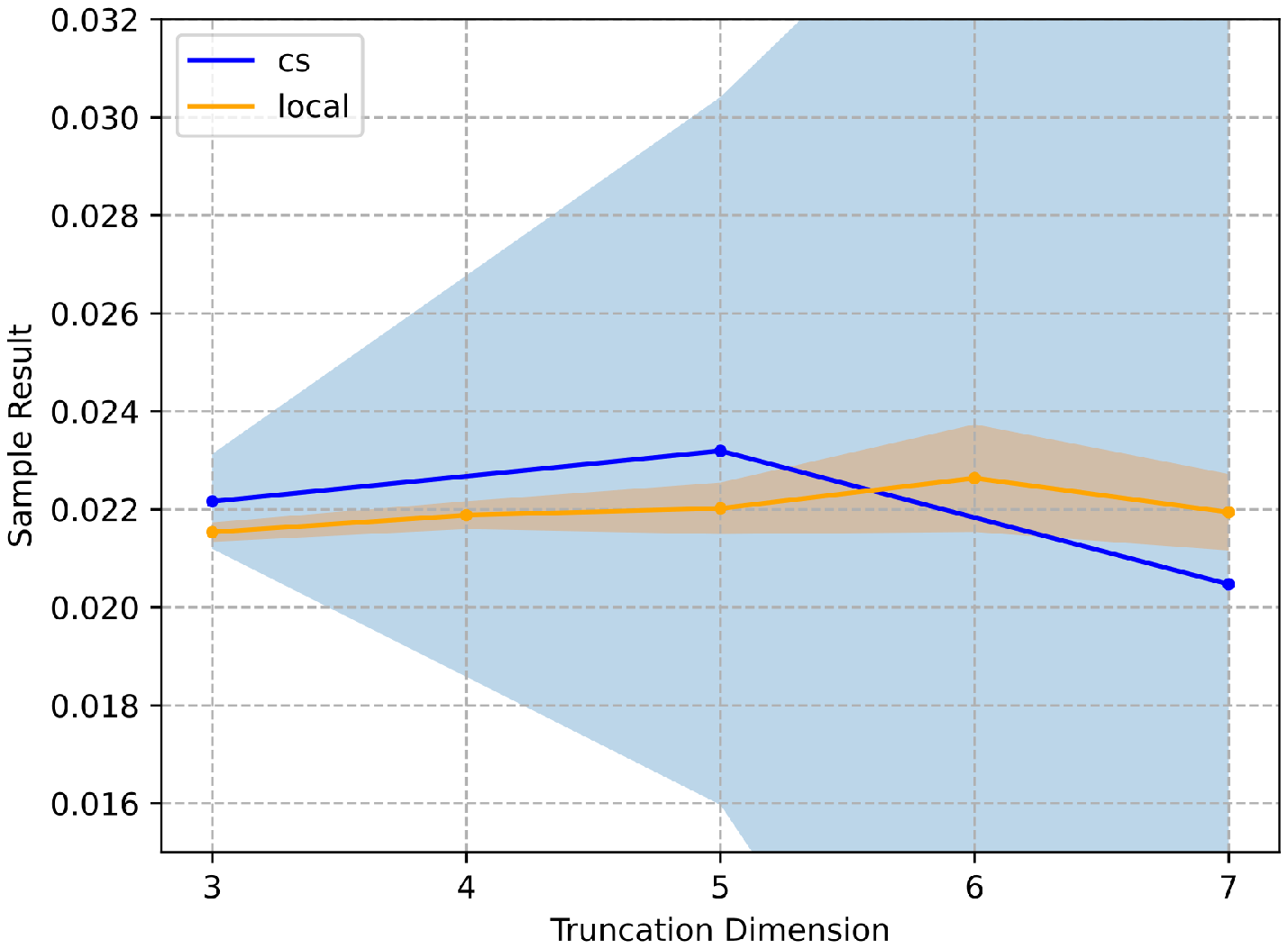}}
    \hspace{15pt}
    \subfigure[]{
        \includegraphics[width=7cm]{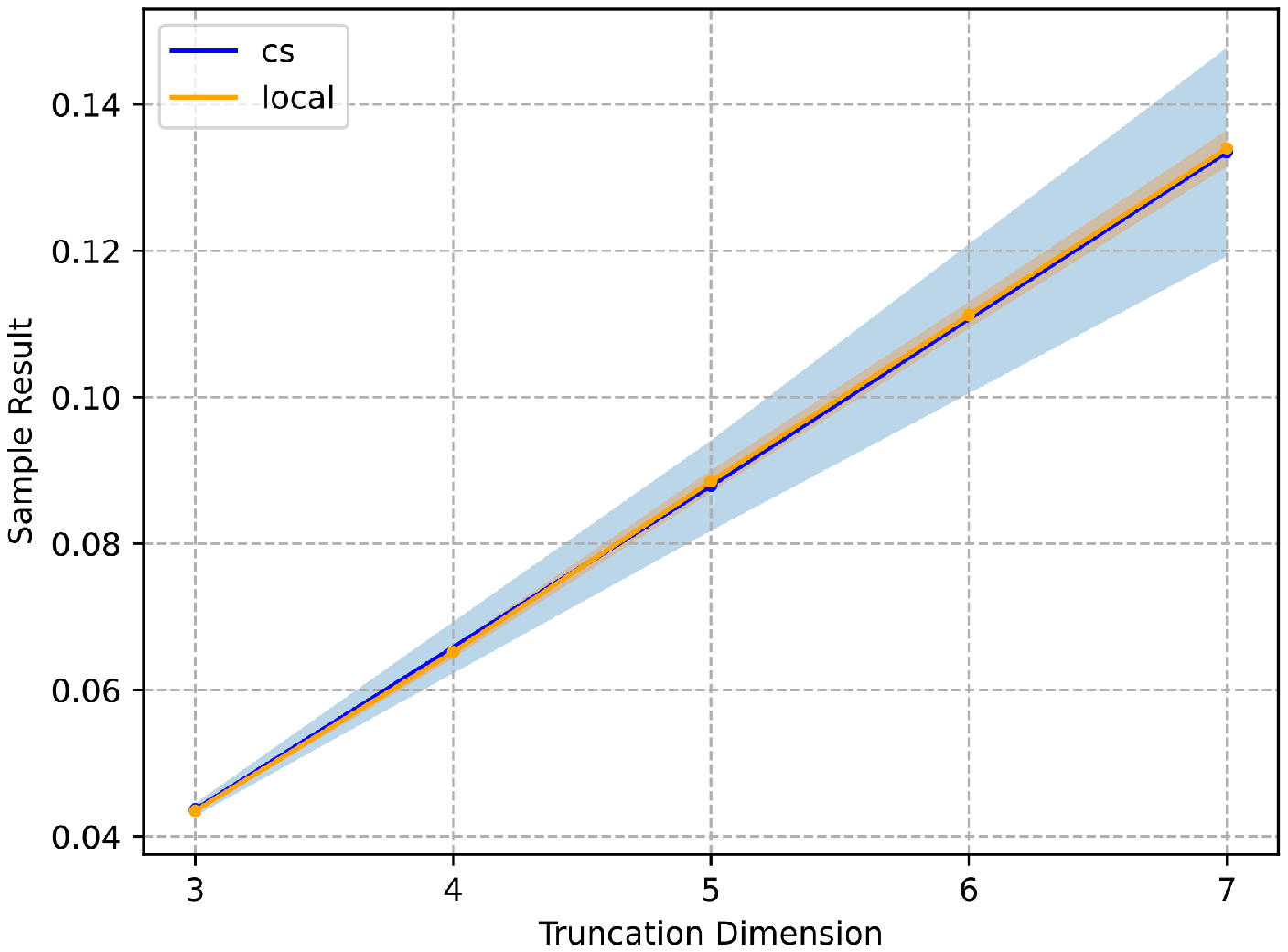}}
\caption{Numerical results for estimating the energy levels of $\rm H_2O$ vibration. (a) estimating the energy expectation of the oscillator's ground state $\ket{000}$ with the two methods (each with $T=10^4$ samples) under different truncated dimensions and their error bands are estimated with $R=10$ repetitions. The exact value of the energy is 0.022. 
The error band for the CS is truncated in (a) since it is too wide for $d>5$.
(b) the results for the energy expectation of the $d$-level GHZ state $\sum_j \ket{jjj} /\sqrt{d}$. Note that the GHZ state changes with the truncation dimension $d$ so the expected value of energy also increases with $d$.}
    \label{fig:numerich2o}
\end{figure}

As shown in Fig.~\ref{fig:numerich2o}, we give the estimation and its errors for the associated algorithms. The result implies the CS method on qudit systems is probably not a good approach to estimate the expectations of a Hamiltonian in the bosonic system due to its huge error. 
The evolution of the standard deviation over the number of samples is given in Fig.~\ref{fig:evolution}. In Fig.~\ref{fig:evolution_ground}, we compare the two methods where we let the truncated dimension be $d = 5$, under which the prediction error due to truncation is negligible for the estimation of ground states. From Fig.~\ref{fig:evolution_qudit}, we can see more clearly that truncation methods bear a much larger variance in exchange for a smaller deviation of expectation by using higher truncated dimensions.

\begin{figure}[htbp]
    \centering
    \subfigure[]{\label{fig:evolution_ground}
    \includegraphics[width=7cm]{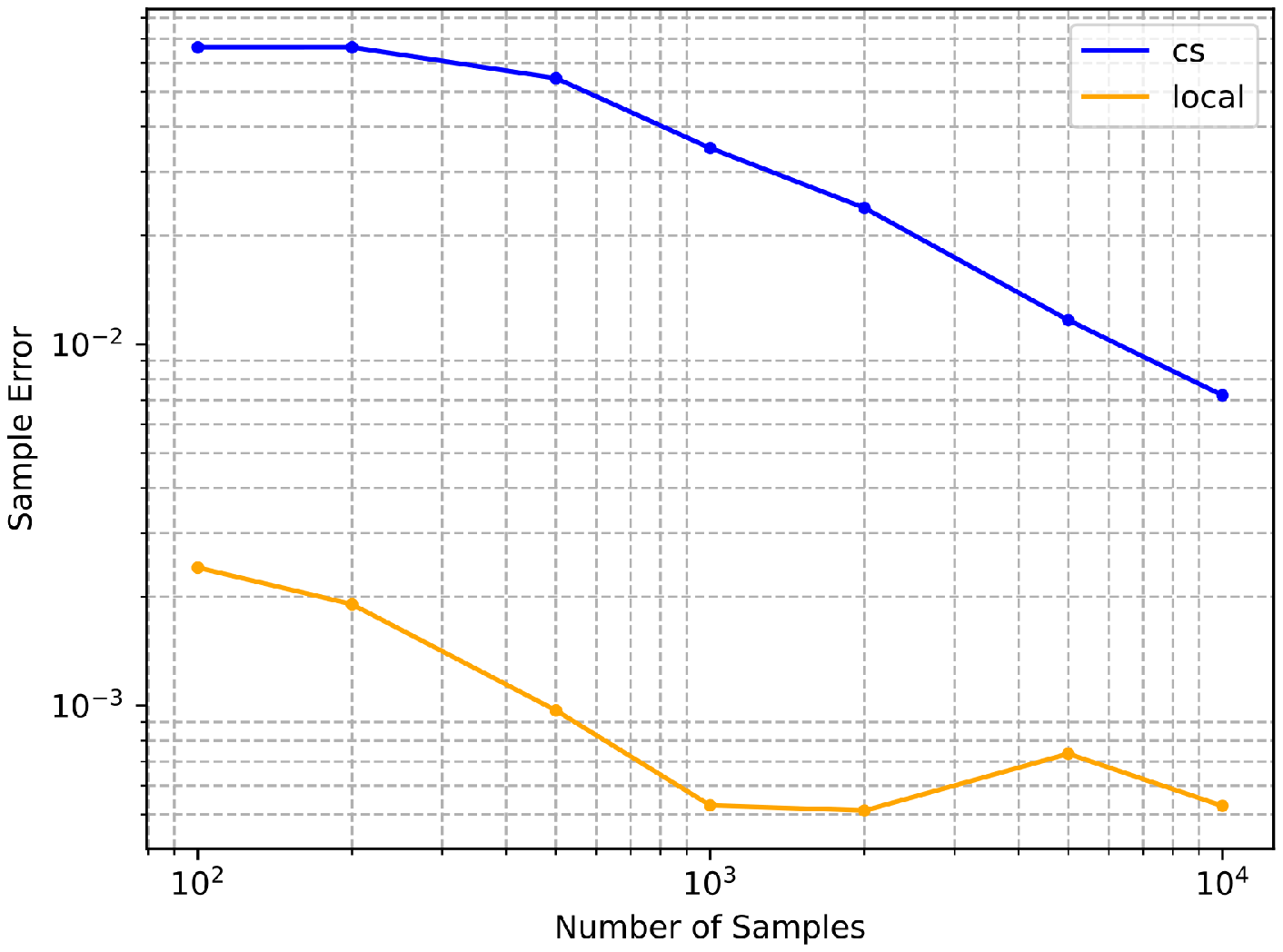}}
    \hspace{15pt}
    \subfigure[]{\label{fig:evolution_qudit}
    \includegraphics[width=7cm]{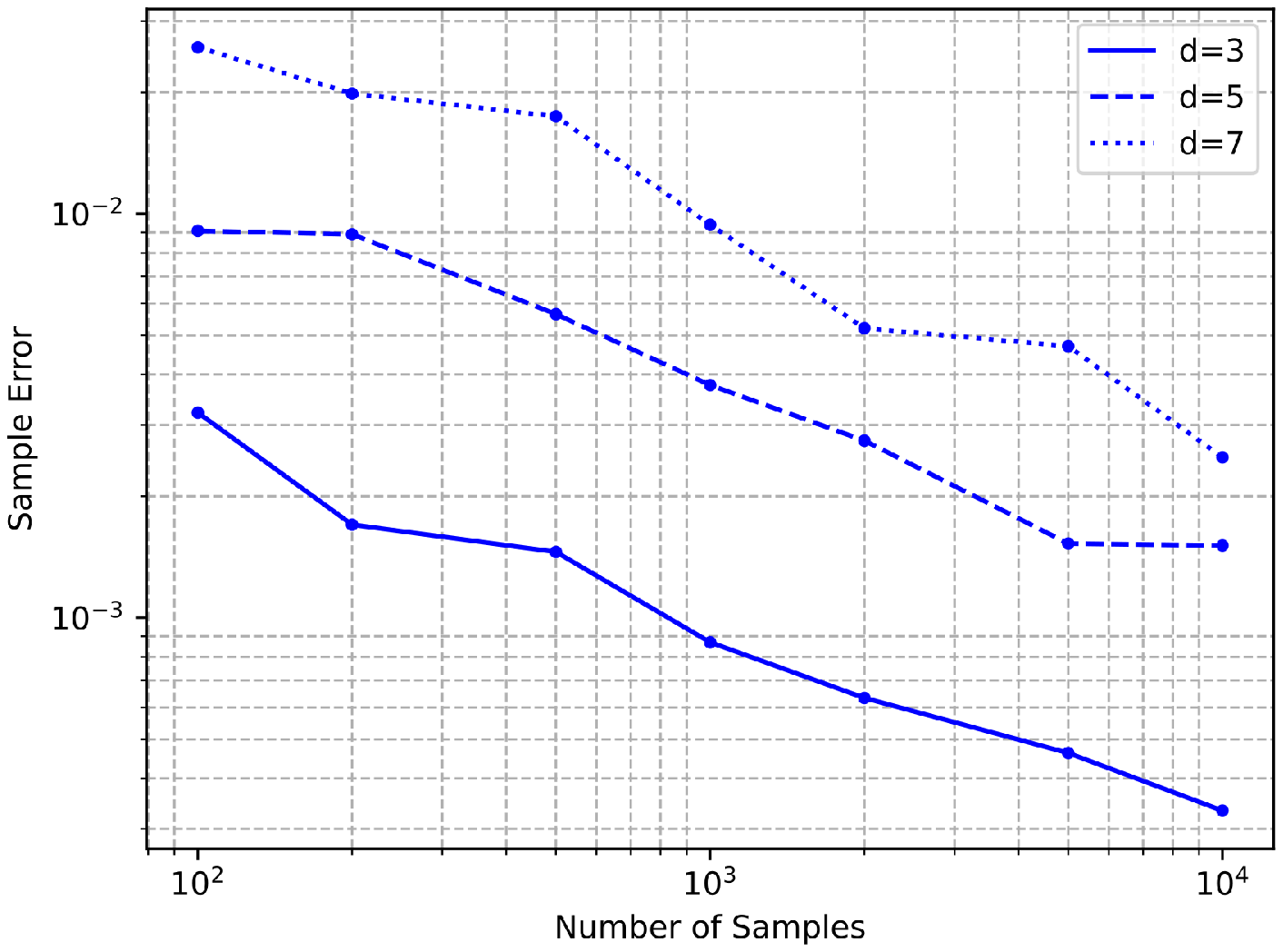}}
    \caption{The evolution of the sampling error over the number of samples $T$ with $R=10$ repetitions. (a) Estimating the oscillator's ground state with the two methods. The truncated dimension $d=5$. (b) Estimating the GHZ state with qudit local measurements.}
    \label{fig:evolution}
\end{figure}

\subsection{Continuous: Optical States}
In the following we consider optical quantum states as the verification of the continuous variable local measurement scheme introduced in Sec.~\ref{sec:BosonContinuous}. 

The trivial verification of unbiasedness of the estimator \eqref{estimator_O} in continuous systems using TMSV states can be found in \ref{Appendix:TMSV_estimation}. In the following, we implement two numerical experiments to show the advantages of the local measurement method in continuous systems. 

At first, we use random multimode Gaussian states to verify our variance bounds~\eqref{Var_o_px}, which is shown in Fig.~\ref{fig:ran_gaussian}. In Fig.~\ref{fig:ran_gaussian_k}, we estimate a set of degree-$K$ random observables with random 10-mode Gaussian states. In the logarithm scale, the error and the degree form a straight line, which is a perfect verification of the exponential relation in  \eqref{Var_o_px}. In Fig.~\ref{fig:ran_gaussian_nk}, we study the relationship between the estimation error and the number of modes with random Gaussian states. The result indicates that the error is nearly independent of the number of modes, and
are mainly determined by the degree $K$. The random observables are normalized such that $\sum_p |\alpha_p|^2 = 1$ for the observables and $\bar x = 0, \trace(V) = 1$ for the states. The estimation error for  $K$-degree observables under $N$-mode states should thus be divided by $(1/\sqrt{N})^K$ for unbiased comparison.

\begin{figure}[htbp]
    \centering
    \subfigure[]{\includegraphics[width=7cm]{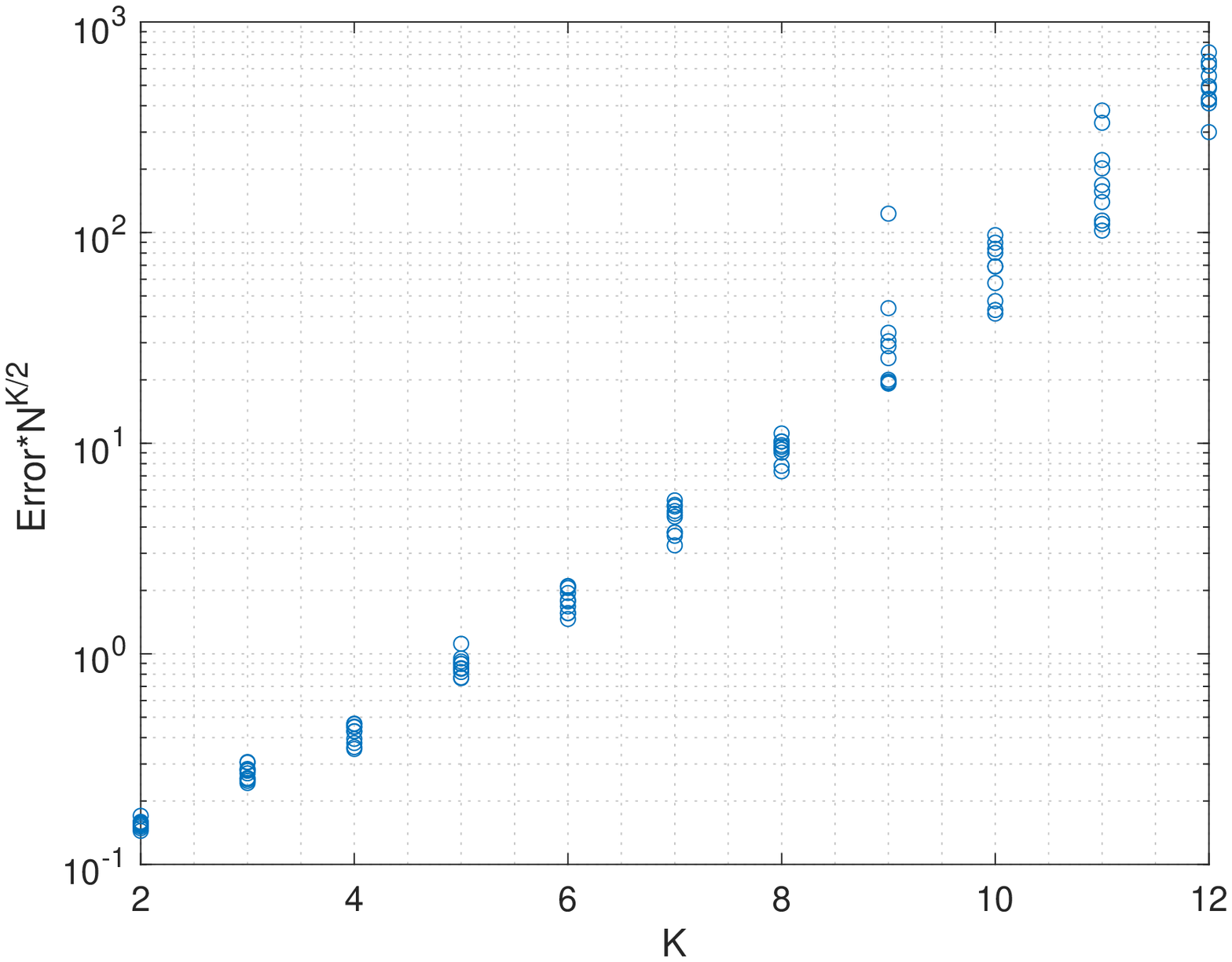}
    \label{fig:ran_gaussian_k}
    }
    \hspace{10pt}
    \subfigure[]{\includegraphics[width=7cm]{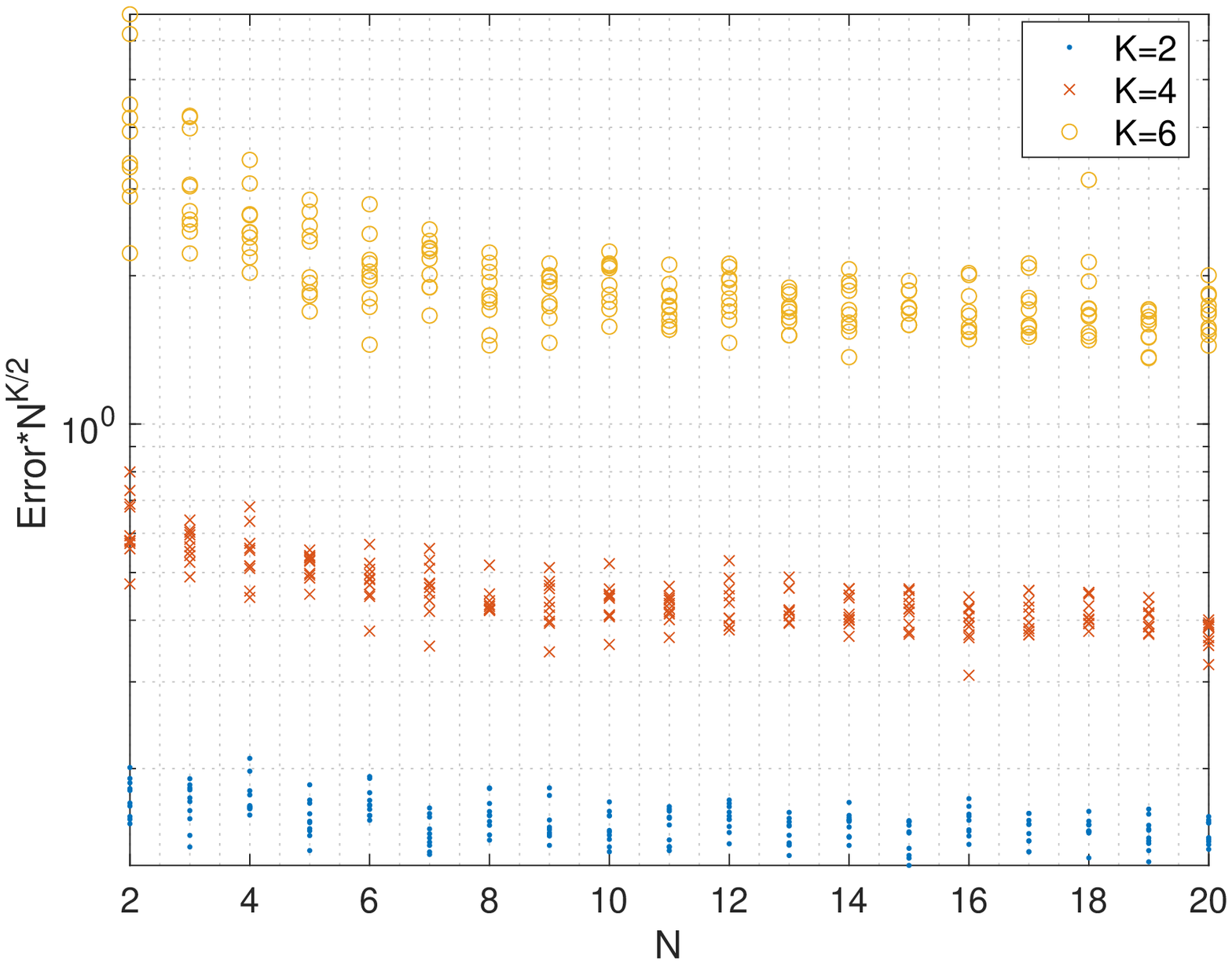}
    \label{fig:ran_gaussian_nk}
    }
    \caption{Estimating the performance of the continuous variable local measurement algorithm using multimode random Gaussian states. (a) the unitless estimation error of random observables with degree $K$. For each degree, observable $O$ with $M=100$ random p-x terms is estimated with $N_{\rm state} = 10$ randomly generated 10-mode Gaussian states with zero mean and normalized covariance matrix. Each data point uses $T=1000$ samples with $R=100$ repetitions. (b) A general picture of estimation error with different number of modes $N$ and degree $K$.}
    \label{fig:ran_gaussian}
\end{figure}

Then we compare the efficiency of different continuous variable local measurement schemes under different number of modes and degrees in Fig.~\ref{fig:squeezedN}. We use an $N$-mode equally squeezed state for testing and estimating on several randomly generated degree-$K$ observables. The result shows that the introducing of overlapped grouping brings us an explicit improvement on the traditional important sampling (is) algorithm, and also yields an advantage over the classical shadow (cs) scheme, especially with a system with a large number of modes and observables with small degrees. 

\begin{figure}[h]
    \centering
    \subfigure[]{\includegraphics[width=7cm]{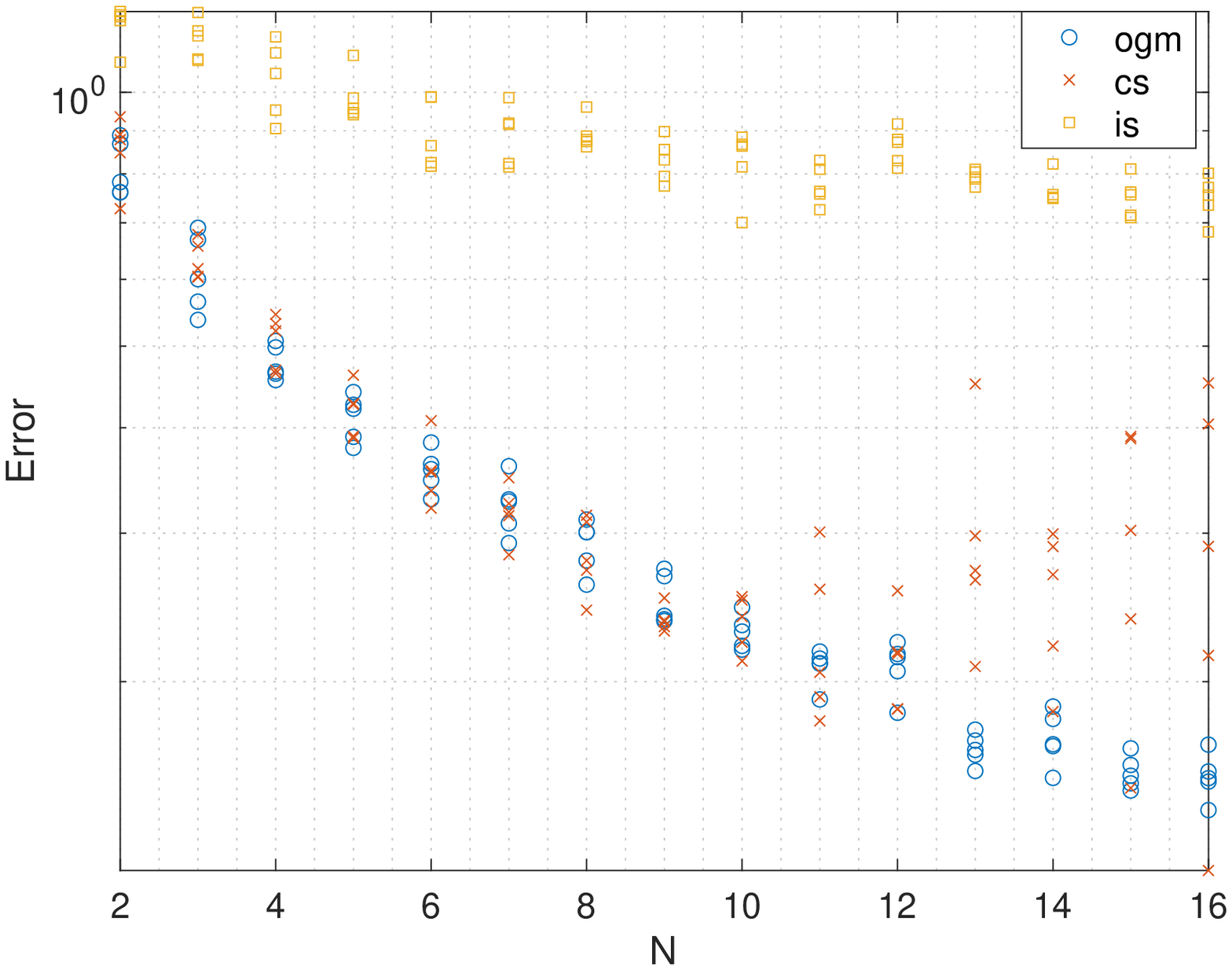}
    \label{fig:squeezedN_is_k2}
    }
    \hspace{10pt}
    \subfigure[]{\includegraphics[width=7cm]{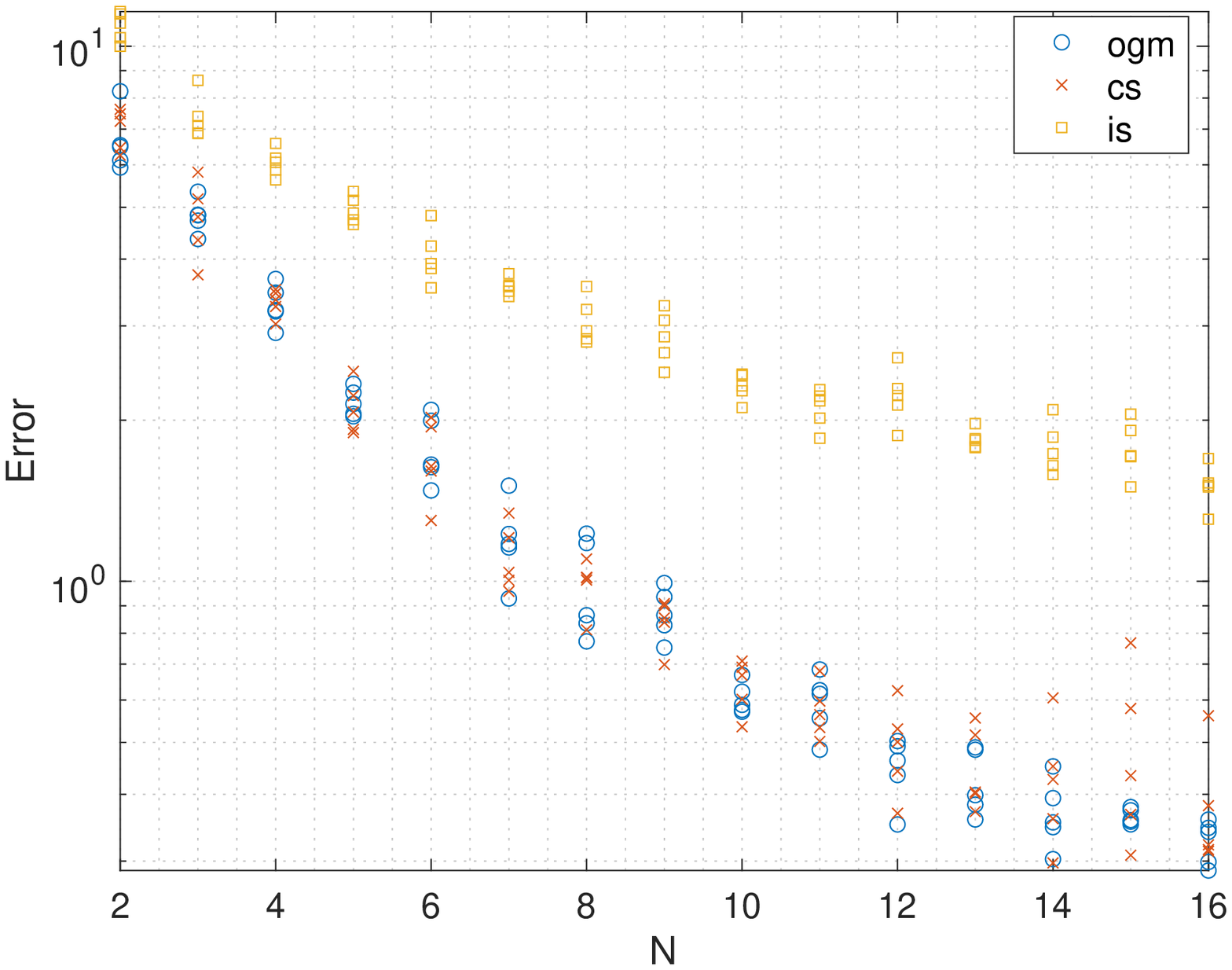}
    \label{fig:squeezedN_is_k4}
    }
    \caption{Estimating the performance of continuous variable local measurements with overlapped grouping optimization (ogm, blue circles), classical shadow scheme (cs, red crosses) and important sampling (is, yellow squares), with equally-squeezed multimode squeezed vacuum states. We estimate 5 normalized observables each with $M=500$ random p-x terms for each figure. The degree $K$ of the observables is (a)$K=2$, and (b)$K=4$.  Each data point is estimated with $T=1000$ independent samplings with $R=100$ repetitions. }
    \label{fig:squeezedN}
\end{figure}

\section{Discussion and outlook}
\label{sec:discussion}
In this work, we discuss different measurement schemes for bosonic systems that are either described by truncated qudits or continuous variable systems, and propose the associated measurement schemes.
We also numerically verify the efficiency of these schemes. In the truncated qudit system, we numerically show the tremendous advantages of overlapped grouping method based on the GGB basis compared to the global cs method for nuclei vibrations. In the continuous variable system, we numerically show that the overlapped grouping method and local cs method both have obvious advantages compared to the traditional importance sampling method.

Our work has wide applications, which is especially useful for estimating the properties of molecular and atomic bosonic systems as well as improving the measurement accuracy and efficiency of optical quantum computers. Other applications include estimating the purity and error of continuous variable quantum systems. We leave an interesting open question that whether there exist global classical shadow variants in truncated qudits and continuous variable systems, with the acceptable bound of variance/error.

\ack{
This work is supported by the National Natural Science Foundation of China (Grant No.~12175003, No.~12147133), and Zhejiang Lab's International
Talent Fund for Young Professionals.
The numerical experiment is supported by the High-performance Computing Platform of Peking University.
}

%\newpage
\appendix 
\newpage
\section{Proof for the variance bound of local measurements}
\label{appendix_variance}
We hereby give a proof for the variance bounds of the general local measurement scheme under the classical shadow scheme (uniform measurement probability), and then specialize it to Proposition \ref{pro:qudit_local} and Proposition \ref{p-x theorem}. The proof is a generalized version of Huang's proof of the variance of classical shadow under Pauli measurements. For the local measurement group of cardinality $D$ with uniform distribution, we decomposing the $k$-local observable $O$ as $O = \sum_p \alpha_p Q_p$, with Eq.(\ref{variance_O}) we have
\begin{equation}
\begin{aligned}
{\rm Var}(\hat o) &= \sum_{p,q}\alpha_p\alpha_q g(Q_p,Q_q)\trace(\rho Q_pQ_q) \le ||\sum_{p,q}\alpha_p\alpha_q g(Q_p,Q_q)Q_pQ_q||_\infty,
\end{aligned}
\label{Var}
\end{equation}
where 
\begin{equation}
\begin{aligned}
g(Q_p,Q_q)&= \dfrac{\sum_{P:Q_p\triangleright P \bigwedge Q_q\triangleright P}\mathcal{K}(P)}{\chi(Q_p)\chi(Q_q)} 
= \dfrac{\dfrac{1}{(D-1)^n}\sum_{P:Q_p\triangleright P \bigwedge Q_q\triangleright P}1}{\dfrac{1}{(D-1)^{2n}}\sum_{P:Q_p\triangleright P}1\sum_{P:Q_q\triangleright P}1},
\end{aligned}
\end{equation}
for general group measurements~\cite{wu2021overlapped}. Note that we choose a uniform probability for all the measurements here as an upper bound of the variance with the optimized probability. In the following, we count the summations. Let $|Q|$ denote the locality of $Q$, which equals the size of set $\cbra{j|Q_j \ne I}$. Since $P$ only varies in qudits/modes $j$ where $Q_j=I$, we have
\begin{equation}
\sum_{P:Q_p\triangleright P}1 = (D-1)^{n-|Q_p|},
\end{equation}
\begin{equation}
\sum_{P:Q_p\triangleright P \bigwedge Q_q\triangleright P}1 = (D-1)^{n-|Q_p\bigcup Q_q|}.
\end{equation}
Set $s = |Q_p\bigcap Q_q|=|Q_p|+|Q_q|-|Q_p\bigcup Q_q|$, we can write $g(Q_p,Q_q) = (D-1)^s$ when $Q_p,Q_q$ are compatible, and $0$ if not.\\

Now returning to Eq. (\ref{Var}). Since our observable is $k$-local, we have
\begin{equation}
\begin{aligned}
&\sum_{p,q}\alpha_p\alpha_q g(Q_p,Q_q)Q_pQ_q \\
&= \sum_{p,q}\alpha_p\alpha_q (D-1)^{|Q_p|+|Q_q|}Q_pQ_q\cdot (D-1)^{k-|Q_p\bigcup Q_q|}\cdot (D-1)^{-k}\\
&= (D-1)^{-k} \sum_{Q_r\in \Q^{\otimes k}}\sum_{Q_p,Q_q\triangleright Q_r}\alpha_p\alpha_q (D-1)^{|Q_p|+|Q_q|}Q_pQ_q \\
&= (D-1)^{-k}\sum_{Q_r\in \Q^{\otimes k}}\left(\sum_{Q_p\triangleright Q_r}\alpha_p (D-1)^{|Q_p|}Q_p\right)^2.
\end{aligned}
\label{sum}
\end{equation}
The second term can be replaced by summing over the identity terms.
From the convexity of the spectral norm, we have
\begin{equation}
\begin{aligned}
&||\sum_{Q_r\in \Q^{\otimes k}}\left(\sum_{Q_p\triangleright Q_r}\alpha_p (D-1)^{|Q_p|}Q_p\right)^2||_\infty \le \sum_{Q_r\in \Q^{\otimes k}}\left(\sum_{Q_p\triangleright Q_r}|\alpha_p| (D-1)^{|Q_p|}||Q_p||_\infty\right)^2.
\end{aligned}
\end{equation}

Using Cauchy-Schwarz:
\begin{equation}
\begin{aligned}
&\sum_{Q_r\in \Q^{\otimes k}}\left(\sum_{Q_p\triangleright Q_r}|\alpha_p| (D-1)^{|Q_p|}||Q_p||_\infty\right)^2\\
&\le  \sum_{Q_r\in \Q^{\otimes k}}\left(\sum_{Q_p\triangleright Q_r}(D-1)^{|Q_p|}\right)
\left(\sum_{Q_p\triangleright Q_r}(D-1)^{|Q_p|}|\alpha_p|^2||Q_p||_\infty^2\right)\\
&\le D^{k}(D-1)^k\sum_p |\alpha_p|^2||Q_p||_\infty^2.
\label{eq:cauchy_sch}
\end{aligned}
\end{equation}
The second inequality of Eq. \eqref{eq:cauchy_sch} holds since there are in total $D^{k}$ number of $Q_r$, and in the second term, each $Q_p$ is calculated $(D-1)^{k-|Q_p|}$ times when summing over $Q_r$. Remembering the factor of $(D-1)^{-k}$ in (\ref{sum}), the final expression is given by
\begin{equation}
{\rm Var}(\hat o) \le D^k\sum_p |\alpha_p|^2 ||Q_p||_\infty^2
\end{equation}

In the qudit case, we use Gell-Mann matrices as the local basis, where $D=d^2$ and the spectral norm is $1$ or $\sqrt{\dfrac{2l}{l+1}}$, 
\rev{hence $||Q_p||_\infty\le (\sqrt{2(d-1)/d})^k$, and
\begin{equation}
\begin{aligned}
\begin{split}
{\rm Var}(\hat o) \le 2^k(d-1)^kd^k\sum_p |\alpha_p|^2 &= 2^k(d-1)^kd^k2^{-k}||O||_2^2 \\
&=(d-1)^kd^{k}||O||_2^2\le (d-1)^kd^{2k}||O||_\infty^2,    
\end{split}
\label{eq:var_a9}
\end{aligned}
\end{equation}}
which gives Proposition \ref{pro:qudit_local}. Eq. \eqref{eq:var_a9} holds since Gell-Mann matrices satisfy $\trace\Lambda_j^\dagger\Lambda_k = 2\delta_{jk}$.

\rev{In the pure p-x string case, the single-mode local measurement basis set $\Q_i^{(j)} \in \{I,x,p\}$, therefore,  we have $D=3$. For each $x$ or $p$ operator, we can choose its upper bound \rev{$B_i^{(j)}$} according to the method in Sec.~\ref{subsec_outmethod_main}. A fixed p-x string $Q_p$ then yields an upper bound of $||Q_p||_\infty \le \prod_{i=1}^n(B_i^{(j)})^{2l_i^{(j)}}$ for $Q_p = \prod_{i=1}^n (Q_i^{(j)})^{l_i^{(j)}}$, which gives us the expression in Theorem \ref{thm:var_px}.}

\section{Extensions and applications}
\label{appendix_expansions_and_application}
In this section we briefly introduce several extensions and applications of the continuous variable local measurement method.

\subsection{Separable observables}
\label{separableobservable}
Here we discuss a special kind of observable that can be separated to the simple summation of the momentum part and the position part. Generally, it is written as
\begin{equation}
O(\bm{p},\bm{x}) = U(\bm{p}) + V(\bm{x}) = U(p_1,\cdots,p_n) + V(x_1,\cdots,x_n),
\label{sep_ob}
\end{equation}
for $n$-mode boson systems. Suppose $U$ and $V$ are expandable, the observable (\ref{sep_ob}) can be easily written in the standard form (\ref{O_px_decomposition})

\begin{equation}
O =  U_0 + \sum_{i=1}^n U_i p_i + \sum_{i,j}U_{ij}p_ip_j + \cdots + 
V_0 + \sum_{i=1}^n V_ix_i + \sum_{i,j}V_{ij}x_ix_j+\cdots.
\end{equation}

Notice that all the momentum terms commute with the p-x string $P_p = p_1p_2\cdots p_n$ and all the potential terms commute with the p-x string $P_x = x_1x_2\cdots x_n$, so using this two measurement basis is already enough for this problem. The estimator is

\begin{equation}
\begin{aligned}
\hat o  = \mathcal{K}(P_p)^{-1}U(P_p^{(1)},P_p^{(2)},\cdots,P_p^{(n)})\delta_{P,P_x}+ \mathcal{K}(P_x)^{-1}V(P_x^{(1)},P_x^{(2)},\cdots,P_x^{(n)})\delta_{P,P_x},
\end{aligned}
\end{equation}

where $P_p^{(i)}$ and $P_x^{(i)}$ denote the $i$th measurement result, $\mathcal{K}(P)$ denotes the probability of choosing this measurement.
Then we estimate the performance of this estimator. Set $\mathcal{K}(P_p) = \lambda$ and $\mathcal{K}(P_x) = 1-\lambda$. We also denote $Q_p^{(J)} = p_{J_1}p_{J_2}\cdots p_{J_n}$ and $Q_x^{(J)} = x_{J_1}x_{J_2}\cdots x_{J_n}$ for $J = (J_1,\cdots,J_n)$. We can then calculate the factor
\begin{equation}
\begin{aligned}
& g(Q_p^{(J)},Q_p^{(K)}) = \frac{\lambda}{\lambda^2} = \frac{1}{\lambda}, \\
& g(Q_p^{(J)},Q_x^{(K)}) = 0, \\
& g(Q_x^{(J)},Q_x^{(K)}) = \frac{1-\lambda}{(1-\lambda)^2} = \frac{1}{1-\lambda}.
\end{aligned}
\end{equation}
According to (\ref{variance_O}), we have
\begin{equation}
\begin{aligned}
{\rm Var}(\hat o) &= 
\frac{1}{\lambda}\sum_{J,K}U_JU_K\trace(\rho Q_p^{(J)} Q_p^{(K)})+
\frac{1}{1-\lambda}\sum_{J,K}V_JV_K\trace(\rho Q_x^{(J)} Q_x^{(K)}) \\
&\le \frac{1}{\lambda}||\sum_{J,K}U_JU_K Q_p^{(J)} Q_p^{(K)}||_\infty +
\frac{1}{1-\lambda}||\sum_{J,K}V_JV_K Q_x^{(J)} Q_x^{(K)}||_\infty \\
&= \frac{1}{\lambda}||U||_\infty^2 + \frac{1}{1-\lambda}||V||_\infty^2,
\end{aligned}
\end{equation}
where $||U||_\infty,||V||_\infty$ are the spectral norms. Optimizing the parameter $\lambda$ and we can get the final bound:
\begin{equation}
{\rm Var}(\hat o) \le (||U||_\infty + ||V||_\infty)^2.
\end{equation}
Setting $U(\bm{p}) = \sum_{i=1}^n p_i^2 / 2$ we get (\ref{bosoninpotential}). 
Sometimes $||U||_\infty$ and $||V||_\infty$ can be very large or even infinity, in which case we would still deal with it by setting the cutting probability. Anyway, it is actually impossible to bound a continuous variable, while in most cases our method works quite well.\\

\subsection{Measuring with noise}
\label{measurewithnoise}
In this section we evaluate the change in result when the continuous measurement has a shift, often caused by precision, that
\begin{equation}
\mu(P)\to \mu'(P) = \mu(P) + e(P),
\end{equation}
with small error $e(P)$, independent of each other and having zero expectation. 
We first consider the change in the predicted expectation of the observable. As
\begin{equation}
\begin{aligned}
&\mathbb{E}_{\mu(P),e(P)}\mu'(P,{\rm supp}(Q)) \\
&= \mathbb{E}_{e(P)}\mathbb{E}_{\mu(P)} [\mu'(P,{\rm supp}(Q))|e(P)] \\
&= \mathbb{E}_{e(P)}\mathbb{E}_{\mu(P)} \prod_{i\in {\rm supp}(Q)}(\mu(P_i) + e(P_i)) \\
&= \mathbb{E}_{e(P)}\trace(\rho\prod_{i\in {\rm supp}(Q)}(P_i + e(P_i))) \\
&= \mathbb{E}_{e(P)}\trace(\rho\left[Q + \sum_{i\in {\rm supp}(Q)}\frac{Q}{P_i}e(P_i) + \sum_{i,j\in {\rm supp}(Q), i<j}\frac{Q}{P_iP_j}e(P_i)e(P_j)+\cdots \right]) \\
&= \trace(\rho Q) + \trace(\rho\left[\sum_{i\in {\rm supp}(Q)}\frac{Q}{P_i}\mathbb{E}_{e(P)}e(P_i) + \sum_{i,j\in {\rm supp}(Q), i<j}\frac{Q}{P_iP_j}\mathbb{E}_{e(P)}e(P_i)e(P_j)+\cdots \right]) \\
&= \trace(\rho Q).
\end{aligned}
\end{equation}
The other terms are 0 because these errors are independent and have an expectation of 0. So the expectation of the measurement result is not affected, indicating that the final expectation is not affected. 

Then we estimate the variance. Similar to the above derivation, we have 
\begin{equation}
\begin{aligned}
&\mathbb{E}_{\mu(P),e(P)}\mu'(P,{\rm supp}(Q)\mu'(P,{\rm supp}(R))\\
&= \mathbb{E}_{e(P)}\mathbb{E}_{\mu(P)} \prod_{i\in {\rm supp}(Q)}(\mu(P_i) + e(P_i)) \prod_{i\in {\rm supp}(R)}(\mu(P_i) + e(P_i))\\
&= \mathbb{E}_{e(P)}\trace(\rho\prod_{i\in {\rm supp}(Q)}(P_i + e(P_i))\prod_{i\in {\rm supp}(R)}(P_i + e(P_i))) \\
&= \mathbb{E}_{e(P)}\trace(\rho\left[QR + \sum_{i\in {\rm supp}(Q)\cup {\rm supp}(R)}\frac{QR}{P_i}e(P_i) + \sum_{i,j\in {\rm supp}(Q)\cup {\rm supp}(R), i<j}\frac{QR}{P_iP_j}e(P_i)e(P_j)+\cdots \right]) \\
&= \trace(\rho QR) + \trace(\rho\left[\sum_{i\in {\rm supp}(Q)\cup {\rm supp}(R)}\frac{QR}{P_i}\mathbb{E}_{e(P)}e(P_i) + \sum_{i,j\in {\rm supp}(Q)\cup {\rm supp}(R), i<j}\frac{QR}{P_iP_j}\mathbb{E}_{e(P)}e(P_i)e(P_j)+\cdots \right]) \\
&\approx \trace(\rho QR) + \trace(\rho\sum_{i,j\in {\rm supp}(Q)\cup {\rm supp}(R), i<j}\frac{QR}{P_iP_j}\mathbb{E}_{e(P)}e(P_i)e(P_j)).
\end{aligned}
\label{Var_noise_1}
\end{equation}
The last row we use the condition that $e(P_i)$ is small so we only expand it to the second level. Now for the extra term, only when $e(P_i) = e(P_j)$ we would get a non-zero expectation:
\begin{equation}
\sum_{i\in {\rm supp}(Q)\cup {\rm supp}(R)}\frac{QR}{P_iP_j}\mathbb{E}_{e(P)}e(P_i)e(P_j) = \sum_{i\in {\rm supp}(Q)\cap {\rm supp}(R)}\frac{QR}{P_i^2}\mathbb{E}_{e(P)}e(P_i)^2 = \sum_{i\in {\rm supp}(Q)\cap {\rm supp}(R)}\frac{QR}{P_i^2}{\rm Var}(e(P_i)).
\label{Var_noise_2}
\end{equation}
With (\ref{v_G_noise})(\ref{Var_noise_1})(\ref{Var_noise_2}) we get the final expression
\begin{equation}
{\rm Var}(\hat o) = \sum_{p,q} \alpha_p\alpha_q g(Q^{(p)},Q^{(q)}) \left(\trace(\rho Q^{(p)} Q^{(q)})+\sum_{i\in {\rm supp}(Q^{(p)})\cap {\rm supp}(Q^{(q)})}\trace(\rho\frac{Q^{(p)}Q^{(q)}}{P_i^2}){\rm Var}(e(P_i))\right).
\label{Var_noise_appendix}
\end{equation}
Note that the variance expression \eqref{Var_noise_appendix} is simply the summation of the original variance $V_{\rm o}$ and an extra variance 
\begin{equation}
    V_e = \sum_{p,q} \alpha_p\alpha_q g(Q^{(p)},Q^{(q)}) \sum_{i\in {\rm supp}(Q^{(p)})\cap {\rm supp}(Q^{(q)})}\trace(\rho\frac{Q^{(p)}Q^{(q)}}{P_i^2}){\rm Var}(e(P_i)).
\end{equation}
Bounding this variance is simple. Supposing the variance of the error also has an upper bound for all modes that $\Var(e(P_i)) \le B_e^2$, for $k$-local system we have
\begin{equation}
    \begin{aligned}
    V_e &\le B_e^2 \sum_{p,q} \sum_{i\in {\rm supp}(Q^{(p)})\cap {\rm supp}(Q^{(q)})} \alpha_p\alpha_q g(Q^{(p)},Q^{(q)}) \trace(\rho\frac{Q^{(p)}Q^{(q)}}{P_i^2})\\
    &= B_e^2 ||\sum_i \sum_{p,q | i\in {\rm supp}(Q^{(p)}\bigcap Q^{(q)})} \alpha_p\alpha_q (D-1)^{|Q^{(p)}|+|Q^{(q)}|+k-|Q^{(p)}\bigcup Q^{(q)}|-k} \frac{Q^{(p)}Q^{(q)}}{P_i^2}||_\infty\\
    &= B_e^2 (D-1)^{-k}||\sum_i \sum_{Q^{(r)} \in \Q^{\otimes k}} \sum_{Q^{(p)},Q^{(q)}\triangleright Q^{(r)}, i\in {\rm supp}(Q^{(p)}\bigcap Q^{(q)})} \alpha_p\alpha_q (D-1)^{|Q^{(p)}|+|Q^{(q)}|} \frac{Q^{(p)}Q^{(q)}}{P_i^2}||_\infty\\
    &= B_e^2 (D-1)^{-k}||\sum_i \sum_{Q^{(r)} \in \Q^{\otimes k}} \left(\sum_{Q^{(p)}\triangleright Q^{(r)}, i\in {\rm supp}(Q^{(p)})} \alpha_p (D-1)^{|Q^{(p)}|} \frac{Q^{(p)}}{P_i}\right)^2||_\infty\\
    &\le B_e^2 (D-1)^{-k}\sum_i \sum_{Q^{(r)} \in \Q^{\otimes k}} \left(\sum_{Q^{(p)}\triangleright Q^{(r)}, i\in {\rm supp}(Q^{(p)})} \alpha_p (D-1)^{|Q^{(p)}|}|| \frac{Q^{(p)}}{P_i}||_\infty\right)^2\\
    &\le B_e^2 (D-1)^{-k}\sum_i \sum_{Q^{(r)} \in \Q^{\otimes k}} \left(\sum_{Q^{(p)}\triangleright Q^{(r)}} \alpha_p (D-1)^{|Q^{(p)}|}|| \frac{Q^{(p)}}{P_i}||_\infty\right)^2\\
    \end{aligned}
\end{equation}
The squared term is exactly the one in Appendix \ref{appendix_variance}, except for the $P_i$ term, which reduced the spectral norm by a factor of $B$. Thus the final expression for the bound of extra variance is given by
\begin{equation}
    V_e \le nB^{2k-2}B_e^2\sum_p |\alpha_p|^2
\end{equation}

\subsection{Estimating the purity of a continuous variable quantum state}
Suppose we want to prepare state $\rho_0$ but we in fact get the state $\rho$, but in this case we do not know the form of the error $E_j$, we can still estimate the purity of our quantum state. 

In particular, we suppose our quantum state is shifted to
\begin{equation}
\rho = (1-p)\rho_0 + p\rho_1,
\end{equation}
where $p$ is some unknown small quantity and $\trace(\rho_0\rho_1) = 0$. We then calculate the purity
\begin{equation}
\trace\rho^2 = (1-p)^2\trace\rho_0^2 + p^2\trace\rho_1^2.
\label{purity_0}
\end{equation}
Using $\rho_0$ as the observer, we are able to estimate the quantity
\begin{equation}
\trace(\rho_0\rho) = (1-p)\trace\rho_0^2.
\label{purity_1}
\end{equation}
With (\ref{purity_0})(\ref{purity_1}), we can directly get
\begin{equation}
\trace\rho^2 = \frac{[\trace(\rho_0\rho)]^2}{\trace\rho_0^2} + p^2\trace\rho_1^2.
\end{equation}
Now we want to neglect the $p^2$ term, so we estimate
\begin{equation}
1-\trace\rho^2 = \frac{\trace\rho_0^2-[\trace(\rho_0\rho)]^2}{\trace\rho_0^2} - p^2\trace\rho_1^2 = 2p-p^2(1+\trace\rho_1^2) = \Theta(p),
\end{equation}
which enables us to neglect the $p^2$ terms. And the final expression is
\begin{equation}
\trace\rho^2 = \frac{[\trace(\rho_0\rho)]^2}{\trace\rho_0^2}.
\end{equation}

\subsection{Estimating continuous errors}

Another application of this technique is to directly find the errors in our states. In the discrete case, if we want to prepare some state $\rho_0$ but unfortunately some known systematic error occurs that we actually get the state
\begin{equation}
\rho = (1-p)\rho_0 + \sum_j p_j E_j\rho_0E_j^\dagger,
\end{equation}
where $\sum_j p_j=p$ is the error rate. We can write it in a simpler form by setting $p_0 = 1-p$ and $E_0 = 1$ and get
\begin{equation}
\rho = \sum_j p_j E_j\rho_0E_j^\dagger.
\end{equation} 
Suppose we have already known the error form $E_j$, what we would like to do is to estimate the error probability $p_i$ by measuring a set of observables $\{M_j\}$ on copies of $\rho$. Since we must suppose these errors can be corrected, we can always carefully choose these observables so that
\begin{equation}
E_j^\dagger M_kE_j = \sum_l c^j_{kl} M_l.
\end{equation}
Measuring these observables and we get the expectation:
\begin{equation}
\trace(\rho M_k) = \sum_j p_j \trace(\rho_0E_j^\dagger M_kE_j) = \sum_{j,l}p_jc^j_{kl}\trace(\rho_0M_l).
\end{equation}
Now we can estimate all the $\trace(\rho M_k)$s precisely, and the coefficients $c^j_{kl}$s and $\trace(\rho_0M_l)$s are already known, so the above equations are linear equations of all the $p_j$, and can be exactly solved as long as the rank of the coefficient matrix is no less than the number of types of possible errors.

In the continuous-variable case, generally the error may depend on some random parameter $\theta$:
\begin{equation}
\rho = \sum_j \int d\theta p_j(\theta)E_j(\theta)\rho_0E_j^\dagger(\theta).
\end{equation}

Supposing we can find a set of observables $\{M_j\}$ that
\begin{equation}
E_j^\dagger(\theta) M_kE_j(\theta) = \sum_l c^j_{kl}(\theta) M_l.
\end{equation}
Then we can estimate
\begin{equation}
\begin{aligned}
\trace(\rho M_k) &= \sum_j \int d\theta p_j(\theta)\trace(\rho_0E_j^\dagger M_kE_j)\\& = \sum_j \int d\theta p_j(\theta)\sum_{l}c^j_{kl}(\theta)\trace(\rho_0M_l) \\&= \sum_j\int d\theta p_j(\theta) d_{jk}(\theta),
\end{aligned}
\label{measure_continuous_error}
\end{equation}
where $d_{jk}(\theta)=\sum_{l}c^j_{kl}(\theta)\trace(\rho_0M_l)$ known continuous function of $\theta$. With this equation, we can estimate properties of the error distribution $p_j(\theta)$. For example, in the p-x case, supposing our quantum state may get a position shift
\begin{equation}
\rho = \int da p(a) e^{ipa}\rho_0 e^{-ipa}.
\end{equation}
Setting observables $M_k = x^k$, as
\begin{equation}
e^{ipa}x^k e^{-ipa} = (x+a)^k.
\end{equation}
(\ref{measure_continuous_error}) gives
\begin{equation}
\trace(\rho x^k) = \int da p(a) \trace(\rho_0 (x+a)^k).
\end{equation}
For $k=1$, we can get the expectation of the position shift $a$
\begin{equation}
\trace(\rho x) = \int da p(a) \trace(\rho_0 (x+a)) = \trace(\rho_0 x) + \mathbb{E}a.
\end{equation}
For $k=2$, we can similarly estimate the variance, which has been enough for the gaussian error case. For higher $k$s, we can get the $k$th moment of $a$. \\

\subsection{Dealing with discrete and continuous case together}

In some models, the observer may contain both continuous variables and discrete variables. For example, the Hamiltonian of a particle in a electromagnetic field may contain both its position, momentum and spin, i.e.

\begin{equation}
O = H = \frac{1}{2}\textbf{p}^2 + V(\textbf{x},\textbf{S}).
\end{equation}

Intuitively, we construct 'mixing strings' similar to Pauli strings, GGB strings and p-x strings, allowing all kinds of variables to be put together. We can thus apply the above method to estimate the observer. In the general case the variance bound is

\begin{equation}
{\rm Var}(\hat o) \le (d^2+2)^k{\rm max}(2^k,B^{2K})\sum_p|\alpha_p|^2,
\end{equation}
for $k$-local $K$-degree observer, where $d$ is the dimension of the discrete variable, $B$ is the bound of continuous variables and $\alpha_p$s are the decomposition coefficients of the observer to the mixing strings. If we can separate the continuous part and the discrete part, the variance bound would reduce to
\begin{equation}
{\rm Var}(\hat o) \le d^{3k}||O^{\rm discrete}||_\infty^2 + 3^kB^{2K}\sum_p|\alpha_p^{\rm continuous}|^2.
\end{equation}

\section{Clifford simulation in higher dimension}
\label{dCliffordsimulation}
In the numerical experiment of boson shadow tomography with global measurements, we have to efficiently perform a Clifford circuit on a classical computer. Here we give the exact algorithm of simulating the $d$-dim Clifford circuit and sampling $d$-dim Clifford operators. They are general versions of the qubit Clifford simulation ~\cite{aaronson2004improved,van2021simple}. The overall complexity of the algorithm is $O({\rm poly}(n))$ for $n$ qudits. 

\subsection{Simulating $d$-dim Clifford circuits}

We first introduce higher dimensional stabilizers(or qudit stabilizers)~\cite{gheorghiu2014standard,hostens2005stabilizer,gunderman2020local}. For an $n$-qudit state $\ket{\phi}$, we define the stabilizers as
\begin{equation}
    \Stab(\ket{\phi}) = \{S\in \Pauli_d^n \ | \ S\ket{\phi} = \ket{\phi} \},
\end{equation}
which uniquely determines the state $\ket{\phi}$. We call such a state a qudit stabilizer state. 
All stabilizers of $\ket{\phi}$ form an Abelian group. Moreover, it can be proved that when $d$ is a prime, the number of generators of $\Stab(\ket{\phi})$ is exactly $n$. We denote a single stabilizer by the tabular form
\begin{equation}
    S := (a_1 \ a_2 \ \cdots \ a_n \  b_1 \ b_2 \ \cdots \ b_n \ c) = \omega^c\cdot X_1^{a_1}Z_1^{b_1} \otimes X_2^{a_2}Z_2^{b_2}\otimes\cdots\otimes X_n^{a_n}Z_n^{b_n},
\end{equation}
with $0\le a_i,b_i,c < d$. As we know, qudit Clifford Group $\mathcal{C}_d^n$ is the normalizer of qudit Pauli group $\Pauli_d^n$, so stabilizers become stabilizers under qudit Clifford transformation. We give all possible single mode and double mode transformation here:
\begin{equation}
    \begin{aligned}
         F X^a Z^b F^\dagger &= \omega^{-ab}X^{-b}Z^a, \\
         P X^a Z^b P^\dagger &= \omega^{a(a+\rho_d)/2}X^aZ^{a+b}, \\
         \CNOT_{12} (X_1^{a_1} Z_1^{b_1} \otimes X_2^{a_2} Z_2^{b_2}) \CNOT_{12}^\dagger &= X_1^{a_1}Z_1^{a_2-b_2} \otimes X_2^{a_1+b_1} Z_2^{b_2},
    \end{aligned}
\end{equation}
as the Clifford group $C_d^n$ is generated by $C_d^n = \langle \omega I,F,P,\CNOT\rangle$. For prime $d$, we write the $n$ generators of $\Stab(\ket{\phi})$ as a $n \times (2n+1)$ matrix:
\begin{equation}
    \begin{pmatrix}
a_{1,1} &  \cdots & a_{1,n} & 
b_{1,1} &  \cdots & b_{1,n} & c_1\\
\vdots  &  \ddots & \vdots  & 
\vdots  &  \ddots & \vdots  & \vdots \\
a_{n,1} &  \cdots & a_{n,n} & 
b_{n,1} &  \cdots & b_{n,n} &  c_n\\
    \end{pmatrix}.
    \label{stabilizer_tabular}
\end{equation}

To efficiently do the measurements, we apply the method introduced in \cite{aaronson2004improved} that add $n$ 'destabilizers' above the stabilizers, so the matrix becomes:
\begin{equation}
    \begin{pmatrix}
    a_{1,1} &  \cdots & a_{1,n} & 
b_{1,1} &  \cdots & b_{1,n} & c_1\\
\vdots  &  \ddots & \vdots  & 
\vdots  &  \ddots & \vdots  & \vdots \\
a_{n,1} &  \cdots & a_{n,n} & 
b_{n,1} &  \cdots & b_{n,n} &  c_n\\

a_{n+1,1} &  \cdots & a_{n+1,n} & 
b_{n+1,1} &  \cdots & b_{n+1,n} & c_{n+1}\\
\vdots  &  \ddots & \vdots  & 
\vdots  &  \ddots & \vdots  & \vdots \\
a_{2n,1} &  \cdots & a_{2n,n} & 
b_{2n,1} &  \cdots & b_{2n,n} &  c_{2n}\\

    \end{pmatrix}
=
\begin{pmatrix}
R_1 \\ \vdots \\ R_n \\ R_{n+1} \\ \vdots \\ R_{2n} \\ R_{2n+1}
\end{pmatrix},
\end{equation}
where $R_1,\cdots, R_n$ describe the destabilizers and $R_{n+1},\cdots, R_{2n}$ describe the stabilizers. Initially the tabular is set as $a_{i,i} = b_{n+i,i} = 1$ for $i = 1,\cdots,n$ with other elements all being 0, which denotes the state $\ket{0\cdots 0}$. Then all the operations in the Clifford circuit map to different transformations of the table:\\

\textbf{$F$ on qudit $k$:} For all $i=1,\cdots,2n$, set $c_i := (c_i - a_ib_i) \mod d$, then set $a_{i,k},\  b_{i,k} := (-b_{i,k})\mod d,\  a_{i,k}$.\\

\textbf{$P$ on qudit $k$:} For all $i = 1,\cdots, 2n$, set $c_i := (c_i + a_i(a_i+\rho_d)/2) \mod d$, then set $b_{i,k} := (a_{i,k} + b_{i,k}) \mod d$.\\

\textbf{$\CNOT$ from qudit $k$ to qudit $l$:} For all $i =1,\cdots,2_n$,  set $a_{i,l}:=(a_{i,l}+a_{i,k})\mod d$ and $b_{i,k} := (b_{i,k} - b_{i,l})\mod d$.\\

\textbf{Measure qudit $k$:} Check whether there exists $p \in \{n+1,\cdots,2n\}$ that $a_{p,k} > 0$. If so, then getting all $d$ possible results are equally random. Otherwise, the result is deterministic.
\begin{itemize}
    \item[(1)] \textbf{Random Case:} For a random measurement result, the state must be changed after the measurement. We first use row $p$ to eliminate other row $i$ in the stabilizer that $a_{i,k}>0$, which is performed by multiplying $R_p^{-a_{i,k}}$ and $R_i^{a_{p,k}}$ and set it as the new $R_i$. After that, set $R_{p-n} := R_p$ and then $a_{p,j} := 0, b_{p,j} := \delta_{pj}$ for $j = 1,\cdots,n$. At last, equally choose $c_p \in \{0,1,\cdots, d-1\}$ as the measurement result.
    
    \item[(2)] \textbf{Deterministic Case:} The only thing we need is to determine the measurement result. This is given by solving $\gamma$ from the equation $R_{n+1}^{\gamma_1} \cdot R_{n+2}^{\gamma_2} \cdot \cdots \cdot R_{2n}^{\gamma_n} = \omega^\gamma Z_k$. By using the destabilizers, we can prove that the $\gamma_i$s are exactly the solution of the linear congruence equation
    $\alpha_i \gamma_i \equiv \beta_i(\mod d)$ with the parameters $\alpha_i = \symp(R_i,R_{i+n})$ and $\beta_i = \symp(R_i,Z_k)$ for $i = 1,\cdots,n$. Here the sympletic inner product $\symp(R_p,R_q) = \sum_j (R_{p,j+n}R_{q,j} - R_{p,j}R_{q,j+n})\mod d$ gives the extra phase. The measurement outcome is then $(-\gamma)\mod d$.
    
    \item[(3)] \textbf{Calculate the expectation of Pauli observable $P$:} Get the expectation of a Pauli observable is simple for stabilizer states. If $P$ does not commute with all the stabilizers the result must be 0. Otherwise, we just replace the $Z_k$ in the deterministic case with $P_k$.

\end{itemize}

\subsection{Sampling $d$-dim Clifford operators}

Another technique is sampling a $d$-dimension Clifford operator, which similar to its qubit counterpart, is done by sampling the result stabilizers from applying the Clifford operator to $X_i$ and $Z_i$ and doing a Gaussian elimination to reconstruct the Clifford operator~\cite{van2021simple}. To describe the detail, we still use the tabular form (\ref{stabilizer_tabular}) to denote the stabilizers. The process is then given below:

\begin{enumerate}
    \item \textbf{Sampling the resulting tabular:} We first sample the tabular of $C(X_i) = CX_iC^\dagger$ and $C(Z_i) = CZ_iC^\dagger$. As Clifford transformation would not change the commutation relation of stabilizers, the limitation here is the result must satisfy the commutation rules $C(X)C(Z) = \omega^{-1}C(Z)C(X)$. Suppose we have successfully sampled $C(X)$ and $C(Z)$ of the first qubit, which are two $n$-dit numbers with random phase $\omega^j,j\in\{0,\cdots,d-1\}$, we use the following process to restore the commutation rules:
    
    \begin{itemize}
        \item[(1)] calculate the sympletic inner product $S = \symp(C(X),C(Z))$ and the sympletic inner product on qudit $j$ $S_j = (C(X)_{j+n}C(Z)_{j} - C(X)_jC(Z)_{j+n})\mod d$, where $j$ is the smallest integer that at least one of $C(Z)_j$ and $C(Z)_{j+n}$ is larger than 0. 
        \item[(2)] As the sympletic inner product we need is $-1$, we need to rearrange $C(Z)_j$ and $C(Z)_{j+n}$ until the new sympletic inner product on this qudit $S_j' = (-1 - (S-S_j))\mod d$. Obviously this linear congruence equation has exactly $d$ solutions because a total of $d^2$ possible values of $C(Z)_j$ and $C(Z)_{j+n}$ uniformly map to $d$ possible values of $S_j$. Therefore, to ensure the uniformity of the sampling we need to uniformly distribute the $d^2$ possible values of $C(X)_j$ and $C(X)_{j+n}$ to all $d$ solutions. A simple method to solve this problem is to choose $C(Z)_j$ and $C(Z)_{j+n}$ as the $k$th solution by order of the magnitude of the combined 2-dit number(or any other order designed), where $k$ is the index of the original $C(Z)_j$ and $C(Z)_{j+n}$ among all $d$ possible values of $C(Z)_j$ and $C(Z)_{j+n}$ that give the sympletic inner product $S_j$ by the same order. 
    \end{itemize}
    
    \item \textbf{Gaussian Elimination:} We use Clifford circuits composed of basic Clifford operators to transform $C(X)$ and $C(Z)$ back to $X$ and $Z$. Complex conjugating these operators then give us the circuit form of the original Clifford operator $C$, which can be directly used in Clifford simulation. The process is given as following:
    
    \begin{itemize}
        \item[(1)] \textbf{Clear $Z$ component of $C(X)$:} For each $j$ that $C(X)_{j+n} > 0$, if $C(X)_j > 0$, we apply $P$ gate on this qudit for $k = \inv(C(X)_j,-C(X)_{j+n})$ times, where $\inv(a,b)$ denotes the solution $x$ satisfying $0\le x<d$ of the linear congruence equation $ax\equiv b (\mod d)$. Otherwise, apply $F$ gate for one time.
        
        \item[(2)] \textbf{Clear $X$ component of $C(X)$:}
        Find $J = \{j<n \ | \ C(X)_j>0\}$. If $J$ has more than one element, for each odd $i<|J|$  we apply the $\CNOT$ gate for $k = \inv(C(X)_{J_i},C(X)_{J_{i+1}})$ times from qudit $J_i$ to qudit $J_{i+1}$. Repeating this process then eliminates all the $X$ components except the first one with circuit depth less than $O(\log_2 n)$. After that, we apply $\CNOT$ gate from qudit $J_1$ to qudit 1 for one time and then from qudit 1 to qudit $J_1$ for $k = \inv(C(X)_{J_1},-C(X)_{J_1})$ times to swap the $X$ component to the first position. 
        
        \item[(3)] \textbf{Clear components of $C(Z)$:} If $C(Z)$ is not at the form of $Z_1^{C(Z)_{n+1}}$(neglecting phase) we need this process. We first use an $F$ gate on the first qudit to protect $C(X)$, then we use the same process as eliminating the components of $C(X)$ to eliminate the components of $C(Z)$. At the end, we apply an $F$ gate on qudit 1 again to restore $C(X)$.
        
        \item[(4)] \textbf{Eliminate the powers and phases:} Now $C(X)$ and $C(Z)$ are at the forms of $X^a$ and $Z^b$(neglecting phase). To transform them back to $X$ and $Z$, we apply the following Clifford operator $C$ on the first qudit:
        \begin{equation*}
            C = P^aFP^bFP^kF, \ k = \inv(d-b,d-1),
        \end{equation*}
        with phases $\omega^\alpha$ and $\omega^\beta$. Finally, we apply the $X$ gate for $(-\alpha)\mod d$ and $Z$ gate for $(-\beta)\mod d$ to eliminate the phases. 
    \end{itemize}
    
    \item \textbf{Finishing the process for $n$ qudits:} The above process successfully restore $C(X_1)$ and $C(Z_1)$ to $X_1$ and $Z_1$. For the succeeding qudits we just need to neglect the qudits that have been restored and do the samplings and eliminations on a smaller system.

\end{enumerate}

\section{Estimating on TMSV states}
\label{Appendix:TMSV_estimation}
In this appendix section we show some additional numerical results on a two mode squeezed vacuum state (TMSV) to verify the unbiasedness of the expectation~\eqref{estimator_O} in continuous variable system. TMSV states are the simplest multimode optical model, which are the continuous counterparts of the discrete EPR states and have wide applications in quantum optical experiments~\cite{weedbrook2012gaussian,braunstein2005quantum}. A TMSV state is generally written as
\begin{equation}
    \ket{r} = S_2(r)\ket{0} = \exp[r(\hat b_1 \hat b_2 - \hat b_1^\dagger \hat b_2^\dagger)/2]\ket{0},
\end{equation}
where $r$ is the squeezing parameter, and $\ket{0}$ is the two mode vacuum state. The TMSV state is a gaussian state that can be efficiently simulated on classical computers. It has zero mean and covariance matrix
\begin{equation}
    \bm{V}_{\rm TMSV} = 
    \begin{pmatrix}
        \nu {I} & \sqrt{\nu^2-1}{Z}\\
        \sqrt{\nu^2-1}{Z} & \nu {I}
    \end{pmatrix},
\end{equation}
where $\nu = \cosh(2r)$, $I$ and $Z$ represent single-qubit identity and Pauli-Z in qubit system respectively. We now use the continuous variable local measurements to estimate the average photon number $\bar n$. To verify the result, theoretically we have 
\begin{equation}
    \bar n = \sinh^2 r.    
\end{equation}

\begin{figure}[htbp]
    \centering
    \subfigure[]{\includegraphics[width=7cm]{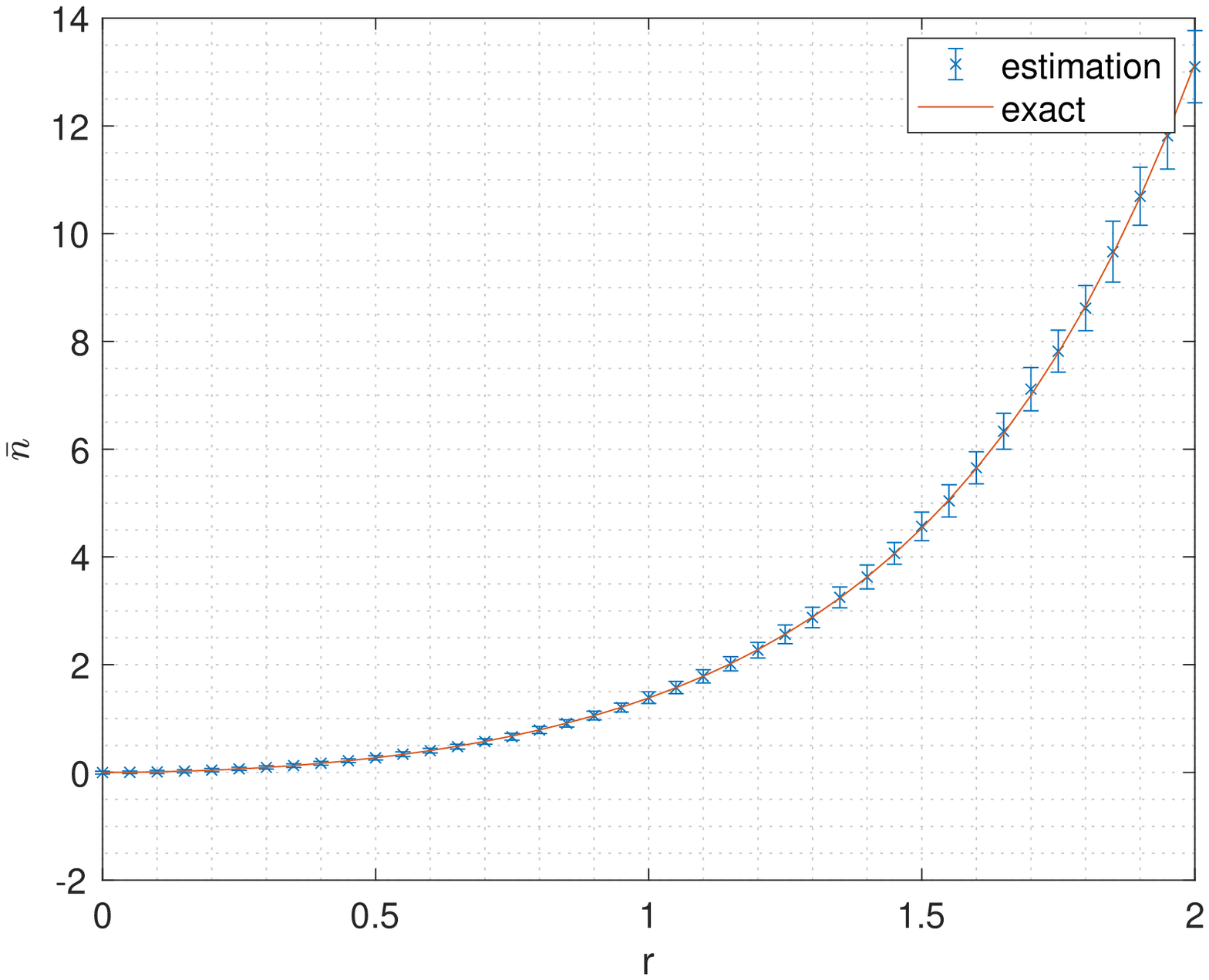}
    \label{fig:TMSV_r}
    }
    \hspace{10pt}
    \subfigure[]{\includegraphics[width=7cm]{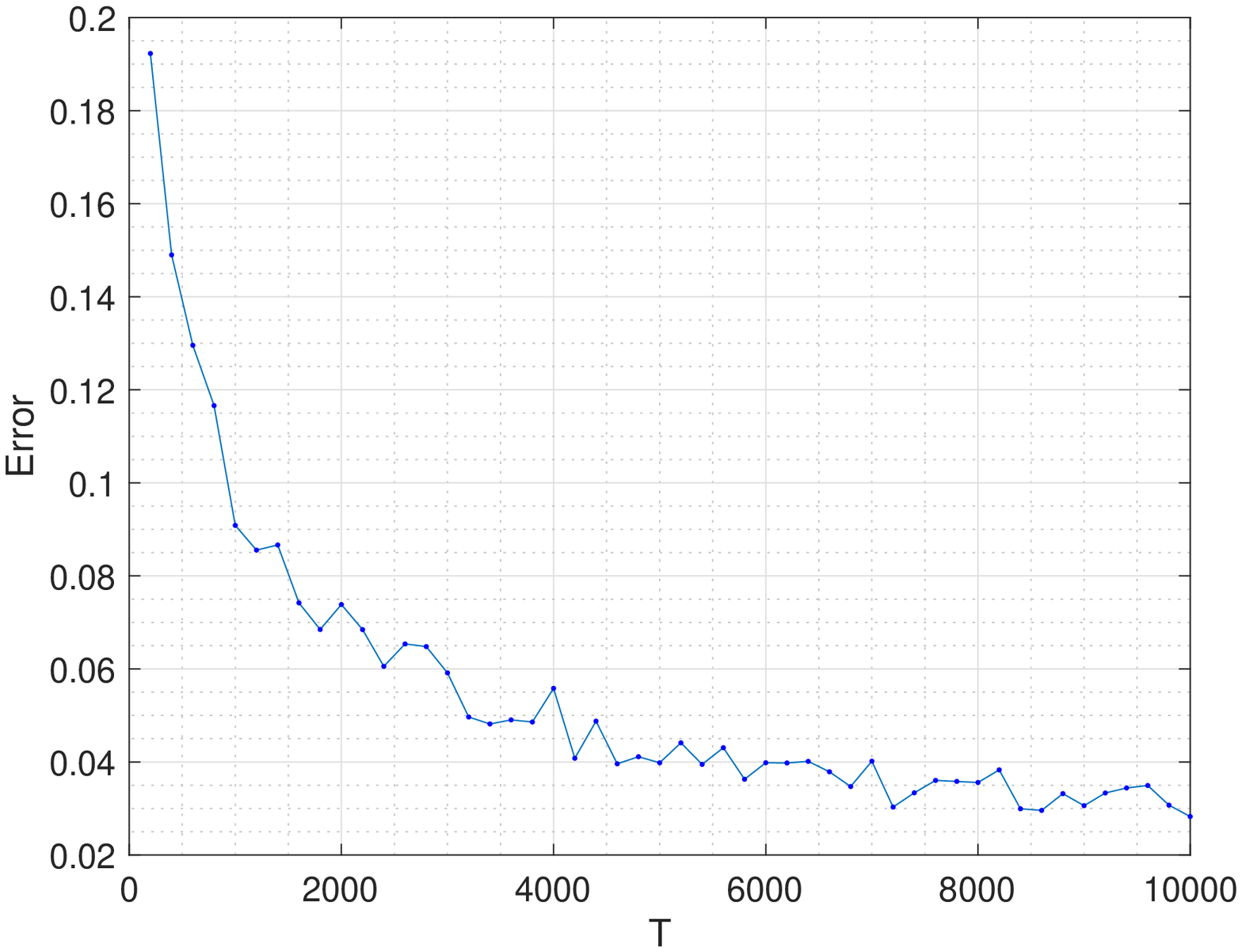}
    \label{fig:TMSV_T}
    }
    \caption{Estimating the average photon number $\bar n$ of the two mode squeezed state.  (a) The estimations and the errors of $\bar n(r)$ under $T=1000$ samples with $R=100$ repeats, where $r$ is the squeezing parameter and the exact value(red line) is given by $\bar n = \sinh^2 r$. (b) The evolution of estimation error corresponding to the number of samples $T$ with $R=100$ repeats. We fix $r=1$ here.}
    \label{fig:TMSV}
\end{figure}

The estimation result and the error are shown in Fig.~\ref{fig:TMSV_r}, which perfectly fits the theoretical result. The precision can be further enhanced by using more samples, see Fig.~\ref{fig:TMSV_T}. 

%\newpage
\section*{References}

\bibliographystyle{unsrt}
\bibliography{refs.bib}

\end{document}